\documentclass[
reprint,onecolumn,
superscriptaddress,
 amsmath,amssymb,
 aps,longbibliography,
prx,nofootinbib,
floatfix,
toc=flat,
]{revtex4-2}

\usepackage{amssymb, amsthm, mathtools}
\DeclareMathOperator{\im}{im}
\DeclarePairedDelimiter{\dpdparens}{\lparen}{\rparen}
\NewDocumentCommand{\parens}{s o m} {
	\IfBooleanTF{#1}
	{\dpdparens{#3}} %
	{\IfNoValueTF{#2} {\dpdparens*{#3}} {\dpdparens[#2]{#3}}}
}
\newcommand{\half}{\ensuremath{\frac{1}{2}}}
\newcommand{\ceil}[1]{\left\lceil#1\right\rceil}
\newcommand{\floor}[1]{\left\lfloor#1\right\rfloor}
\newcommand{\pmat}[1]{\begin{pmatrix}#1\end{pmatrix}}

\usepackage[shortlabels]{enumitem}
\usepackage{multirow}
\usepackage[margin=1in]{geometry}
\usepackage{graphicx} %

\usepackage{verbatim}

\usepackage{tikz}
\usetikzlibrary{shapes.geometric, quantikz2}

\usepackage{algpseudocodex}
\usepackage{algorithm}
\usepackage[dvipsnames]{xcolor}
\makeatletter
\@ifpackageloaded{hyperref} {} {
\usepackage[colorlinks=true, linkcolor=RoyalBlue,citecolor=RoyalBlue,urlcolor=RoyalBlue]{hyperref} %
}
\makeatother

\newtheorem{theorem}{Theorem}
\newtheorem{lemma}{Lemma}
\newtheorem{corollary}{Corollary}

\theoremstyle{definition}
\newtheorem{definition}{Definition}

\theoremstyle{remark}

\DeclareMathOperator{\CNOT}{CNOT}

\DeclareMathOperator{\fSWAP}{fSWAP}
\DeclareMathOperator{\SWAP}{SWAP}

\DeclareMathOperator{\CZ}{CZ}
\DeclareMathOperator{\MOD}{MOD}

\DeclareMathOperator{\wt}{wt}
\DeclareMathOperator{\hd}{hd}

\newcommand{\hc}{\mathrm{h.c.}}
\newcommand{\JW}{\mathrm{JW}}

\newcommand*{\bigo}[1]{O\parens{#1}}

\begin{document}

\title{Simulating fermions with exponentially lower overhead without ancillas}
\title{Low-depth fermion routing without ancillas}
\title{Low-depth fermion routing in product-preserving ternary-tree encodings}
\title{Low-depth fermion routing without ancillas}

\author{Nathan~Constantinides$^{*}$}
\affiliation{Joint Center for Quantum Information and Computer Science, NIST/University of Maryland, College Park, Maryland 20742, USA}
\affiliation{Joint Quantum Institute, NIST/University of Maryland, College Park, Maryland 20742, USA}
\author{Jeffery~Yu$^{*}$}
\affiliation{Joint Center for Quantum Information and Computer Science, NIST/University of Maryland, College Park, Maryland 20742, USA}
\affiliation{Joint Quantum Institute, NIST/University of Maryland, College Park, Maryland 20742, USA}
\author{Dhruv~Devulapalli}
\affiliation{Joint Center for Quantum Information and Computer Science, NIST/University of Maryland, College Park, Maryland 20742, USA}
\affiliation{Joint Quantum Institute, NIST/University of Maryland, College Park, Maryland 20742, USA}
\author{Ali~Fahimniya}
\affiliation{Joint Center for Quantum Information and Computer Science, NIST/University of Maryland, College Park, Maryland 20742, USA}
\affiliation{Joint Quantum Institute, NIST/University of Maryland, College Park, Maryland 20742, USA}
\author{Luke Schaeffer}
\affiliation{Joint Center for Quantum Information and Computer Science, NIST/University of Maryland, College Park, Maryland 20742, USA}
\affiliation{Institute for Quantum Computing, University of Waterloo}
\author{Andrew~M.~Childs}
\affiliation{Joint Center for Quantum Information and Computer Science, NIST/University of Maryland, College Park, Maryland 20742, USA}
\affiliation{Department of Computer Science and Institute for Advanced Computer
Studies, University of Maryland, College Park, Maryland 20742, USA}
\author{Michael~J.~Gullans}
\affiliation{Joint Center for Quantum Information and Computer Science, NIST/University of Maryland, College Park, Maryland 20742, USA}
\affiliation{Department of Computer Science and Institute for Advanced Computer
Studies, University of Maryland, College Park, Maryland 20742, USA}
\author{Alexander~Schuckert$^{\dagger}$}
\affiliation{Joint Center for Quantum Information and Computer Science, NIST/University of Maryland, College Park, Maryland 20742, USA}
\affiliation{Joint Quantum Institute, NIST/University of Maryland, College Park, Maryland 20742, USA}
\author{Alexey~V.~Gorshkov$^{\dagger}$}
\affiliation{Joint Center for Quantum Information and Computer Science, NIST/University of Maryland, College Park, Maryland 20742, USA}
\affiliation{Joint Quantum Institute, NIST/University of Maryland, College Park, Maryland 20742, USA}

\date{\today}
\def\thefootnote{*}\footnotetext{These authors contributed equally to this work.}
\def\thefootnote{$\dagger$}\footnotetext{These authors contributed equally to the supervision of this work. \href{mailto:gorshkov@umd.edu}{gorshkov@umd.edu}, \href{mailto:schuckertalexander@gmail.com}{schuckertalexander@gmail.com}}

\begin{abstract}
Routing is the task of permuting qubits in such a way that quantum operations can be parallelized maximally, given constraints on the hardware geometry. When simulating fermions in the Jordan-Wigner encoding with qubits, a one-dimensional nearest-neighbor-connected geometry is effectively imposed on the system, independently of the underlying hardware, which means that naively, an $O(N)$ depth routing overhead is incurred. Recently, Maskara et al.~[\href{https://arxiv.org/abs/2509.08898}{arXiv:2509.08898}] demonstrated that this routing overhead can be reduced to $O(\log N)$ by decomposing general fermion routing into $O(\log N)$ interleave permutations of depth $O(1)$, using $\Theta(N)$ ancillary qubits and employing measurements and feedforward. Here, we exhibit an alternative construction that achieves the same asymptotic performance. 
We also generalize the result in two ways. Firstly, we show that fermion routing can be performed in depth $O(\log^2 N)$ \emph{without} ancillas, measurements, or feedforward. Secondly, we construct  efficient mappings with $O(\log^2 N)$ depth between all product-preserving ternary tree fermionic encodings, thereby showing that fermion routing in any such encoding can be done efficiently. While these results assume all-to-all connectivity, they also imply upper bounds for fermion routing in devices with limited connectivity by multiplying the fermion routing depth by the worst-case qubit routing depth.

\end{abstract}
\maketitle

\section{Introduction}
\label{sec:intro}

Standard quantum computers based on qubit architectures suffer from locality constraints. These constraints determine which pairs of qubits are allowed to interact, and which interactions may be implemented in parallel. Hence, these constraints present a challenge for performing quantum computation, and must be accounted for in order to implement quantum information processing tasks. One way to handle locality constraints is by performing routing, which is the task of implementing permutations of qubits. By permuting qubits, one may simulate all-to-all connectivity and therefore implement arbitrary quantum circuits with an overhead in circuit depth. In order to minimize this overhead, the problem of routing efficiently under interaction constraints has been studied using swap gates \cite{sivarajah2020,Childs_Schoute_Unsal_2019, Cowtan_2019}, teleportation and mid-circuit measurements \cite{plathanam_2025, padda_2024, devulapalli_2024}, and dynamic reconfigurability \cite{Moses_2023, constantinides_optimal_2024}. In particular, the worst-case depth of routing of qubits arranged in a one-dimensional line using reconfigurable neutral-atom hardware has been shown to be $O(\log N)$ ~\cite{xu2024constantoverhead, constantinides_optimal_2024}.

When simulating fermionic systems on qubit-based quantum computers, routing arises naturally due to the structure of fermion-to-qubit mappings. In the Jordan-Wigner encoding, each fermionic mode is represented by a qubit. To reproduce fermionic statistics, the qubits are arranged on a one-dimensional line, independent of the actual geometry of the Hamiltonian: when two fermions are exchanged between neighboring lattice sites on the Jordan-Wigner line, the wavefunction is given a minus sign. This effectively imposes a one-dimensional nearest-neighbor connectivity constraint, regardless of the underlying hardware geometry. To account for this constraint, fermion routing is typically performed using fermi-SWAP networks~\cite{kivlichan_quantum_2018}, in which layers of fSWAP=SWAP$\cdot$CZ gates permute fermionic modes, where the CZ gate applies the minus sign when two modes are occupied by a fermion. Because of this one-dimensional nearest-neighbour structure, a general routing task on $N$ modes uses $O(N)$ layers of such gates. While 1D $O(\log N)$-depth qubit routing schemes can in principle be directly used to  decompose a general fermion permutation into $O(\log N)$ steps consisting of  interleaves (riffle shuffles)~\cite{xu2024constantoverhead,maskara2025} or staircase permutations (in-order swaps)~\cite{constantinides_optimal_2024}, this does not directly improve the overall asymptotic scaling due to the $O(N^2)$ CZ gates which need to be applied in each interleave or staircase, which naively require $O(N)$ depth. Recently, Maskara et al.~\cite{maskara2025} showed that the interleave CZ circuit can in fact be compressed to depth $O(1)$ by using $\Theta(N)$ ancillas and measurement and feedforward, implying that an arbitrary permutation can be implemented in depth $O(\log N)$.

In this work, we show that staircase permutations can be compressed to depth $O(\log N)$ without ancillas, measurement, or feedforward.\footnote{By using the $O(\log N)$-depth compression of CNOT ladders from Ref.~\cite{Remaud_2025}, the interleave permutations in Ref.~\cite{maskara2025} can also be implemented without ancillas, measurement, or feedforward in $O(\log N)$ depth~\cite{maskaraprivate}.} This means that without ancillas, we achieve depth $O(\log^2 N)$ for arbitrary fermion routing. In addition, we generalize this result to arbitrary product-preserving ternary-tree encodings by proving that there exists a depth $O(\log^2 N)$ mapping between these encodings. Finally, we discuss consequences for simulation of fermionic time evolution.

More specifically, we show the following theorem:

\begin{theorem}
\label{thm:main}
For any product-preserving ternary tree encoding of fermions into $N$ qubits,
there exists a circuit of depth $O(\log^2 N)$ that implements any given fermionic permutation of $N$ fermionic modes.
\end{theorem}

We show this theorem by explicitly constructing the circuit, also enabling a practical implementation. In addition, we show the following corollary, using Ref.~\cite{maskara2025}:

\begin{corollary}
For any product-preserving ternary tree encoding of fermions into $N$ qubits,
there exists a circuit of depth $O(\log N)$ with $O(N)$ ancillas and measurement and feedforward that implements any given fermionic permutation of $N$ fermionic modes. 
\end{corollary}

The remainder of this manuscript is organized as follows. In Section \ref{sec:prelim}, we introduce definitions and notation. In Section \ref{sec:jw-circuit}, we present a $O(\log^2 N)$-depth circuit using staircase permutations that performs any fermionic permutation under the Jordan-Wigner mapping. In Section \ref{sec:ternary-transforms}, we present a $O(\log^2 N)$-depth circuit that converts between any two product-preserving ternary tree mappings. In Section \ref{sec:applications}, we discuss applications to quantum simulation of fermionic models. In Section \ref{sec:outlook}, we briefly discuss some directions for future work.

\section{Preliminaries \label{sec:prelim}}

In this section, we introduce definitions and notation that are used throughout the paper. 

We work with fermions in second-quantized form, where states represent the occupation of a given mode. Given $N_f$ ordered fermionic modes, we denote
a Fock state as $\ket{x_0, \dots, x_{N_f-1}}_f$, where $x_k \in \{0,1\}$ is the number of fermions occupying the $k$th mode. The subscript $f$ distinguishes a Fock state from a qubit state, which has no subscript. %

We denote by $\mathcal{H}_f$ the Hilbert space spanned by the Fock states and $\mathcal{H}_q$ the standard Hilbert space on $N_q$ qubits.

\begin{definition}[Fermion-to-qubit map]
A fermion-to-qubit mapping %
is an isometry $\phi \colon \mathcal{H}_f \to \mathcal{H}_q$.
\end{definition}

Operators in $\mathcal{H}_f$ can be expressed by fermionic annihilation and creation operators $a_k$ and $a_k^\dagger$, respectively. These are defined on Fock states as
\begin{subequations}
\begin{align}
	a_k \ket{x_0, \dots, x_k=0, \dots, x_{N_f-1}}_f &= 0, \\
	a_k \ket{x_0, \dots, x_k=1, \dots, x_{N_f-1}}_f &= (-1)^{\sum_{j=0}^{k-1} x_j} \ket{x_0, \dots, x_k=0, \dots, x_{N_f-1}}_f,  \\
	a_k^\dagger \ket{x_0, \dots, x_k=0, \dots, x_{N_f-1}}_f &= (-1)^{\sum_{j=0}^{k-1} x_j} \ket{x_0, \dots, x_k=1, \dots, x_{N_f-1}}_f,  \\
	a_k^\dagger \ket{x_0, \dots, x_k=1, \dots, x_{N_f-1}}_f &= 0,
\end{align}
\label{eq:fermion-ops-def}
\end{subequations}
and satisfy the anticommutation relations
\begin{equation}
	\{a_j, a_k\} = \{a_j^\dagger, a_k^\dagger\} = 0, \qquad \{a_j, a_k^\dagger\} = \delta_{jk}.
\end{equation}
The set $\{a_k, a_k^\dagger\}_{k=0}^{N_f-1}$ spans the algebra of fermionic operators $\mathcal{B}(\mathcal{H}_f)$. Qubit operators can be expressed in terms of the Pauli basis $\{I_k, X_k, Y_k, Z_k\}_{k=0}^{N_q-1}$.

Given a fermion-to-qubit mapping $\phi$, each fermionic operator $\mathcal{O}_f \in \mathcal{B}(\mathcal{H}_f)$ has a unique corresponding qubit operator $\mathcal{O}_q \in \mathcal{B}(\im \phi)$ satisfying
\begin{equation}
    \label{eq:operator-transform}
	\mathcal{O}_q \phi(\ket{\psi}_f) = \phi \parens{\mathcal{O}_f \ket{\psi}_f}
\end{equation}
for every $\ket{\psi}_f \in \mathcal{H}_f$. That is, $\mathcal{O}_q = \phi \circ \mathcal{O}_f \circ \phi^{-1}$ acts on the mapped qubit states in the same manner that $\mathcal{O}_f$ acts on fermionic states. 
We denote by $\Phi \colon \mathcal{B}(\mathcal{H}_f) \to \mathcal{B}(\im \phi)$ the operator mapping $\mathcal{O}_f \mapsto \mathcal{O}_q$ induced by $\phi$.
Throughout this paper, we consider the case $N_f = N_q = N$. Then every fermion-to-qubit mapping $\phi$ is surjective, i.e., $\im \phi \cong \mathcal{H}_q$, and $\phi$ is a unitary map. We note that in the case $N_q > N_f$, superfast encodings are known which simulate Hamiltonians with bounded degree interaction graphs with $O(1)$ overhead. Namely, the Bravyi-Kitaev superfast encoding \cite{bravyi_kitaev_2002} satisfies this property with $N_q$ scaling with the degree of the interaction graph $d$, $N_q = \Omega(d N_f)$.

For example, the classic Jordan-Wigner mapping~\cite{jordanwigner}, $\phi_{\JW}$, is given by simply treating the list of occupation numbers as a computational basis state in the qubit space:
\begin{equation}
	\phi_{\JW}\parens[\big]{\ket{x_0, x_1, \dots, x_{N-1}}_f} = \ket{x_0, x_1, \dots, x_{N-1}}. \label{eq:jw}
\end{equation}
Equation (\ref{eq:jw}) induces a mapping of creation and annihilation operators to qubit operators as follows:
\begin{subequations}
\begin{align}
	\Phi_{\JW}(a_k) &= \half Z_0 \otimes Z_1 \otimes Z_2 \otimes \dots \otimes Z_{k-1} \otimes (X_k + i Y_k), \\
	\Phi_{\JW}(a_k^\dagger) &= \half Z_0 \otimes Z_1 \otimes Z_2 \otimes \dots \otimes Z_{k-1} \otimes (X_k - i Y_k).
\end{align}
\end{subequations}
Note that $\sigma_k^+ = \half \parens{X_k +i Y_k}$ and $\sigma^-_k = \half \parens{X_k - i Y_k}$ are simply the qubit lowering and raising operators, respectively, while the leading string of Pauli $Z$s implements the $(-1)^{\sum_{j=0}^{k-1} x_j}$ phase appearing in Eq.~\eqref{eq:fermion-ops-def}.

The Jordan-Wigner mapping is often considered the simplest fermion-to-qubit mapping due to its simplicity in mapping the states. However, its disadvantage is in the linear weight of the Pauli strings of mapped fermionic operators~\cite{Seeley_2012}. Consequently, many other mappings %
(see e.g.~Refs.~\cite{bravyi_kitaev_2002, verstraete2005mapping, Setia_2019, Jiang_2020, derby2021})
have been developed. For example, the Bravyi-Kitaev transformation~\cite{bravyi_kitaev_2002} maps individual fermion operators to log-weight Pauli strings.
However, as argued in Section~\ref{sec:intro}, this should not be the sole figure of merit for a fermion-to-qubit map. Rather, the overall simulation cost depends on the extent to which multiple Hamiltonian terms may be implemented in parallel. 
Therefore, in this work, we focus on the task of implementing fermion permutations.

\begin{definition}[Fermion permutation]
\label{def:fermion-permutation}
Let $\sigma \colon [N] \to [N]$ be a permutation. The fermion routing operator implementing $\sigma$ is the operator $U_\sigma$ satisfying %
\begin{subequations}
\begin{align}
	U_\sigma a_k U_\sigma^\dagger &= a_{\sigma(k)} \quad \forall\, k \in [N] \label{eq:fermion-permutation-a} \\
	U_\sigma \ket{0 \ldots 0}_f &= \ket{0 \ldots 0}_f. \label{eq:fermion-permutation-b}
\end{align}
\end{subequations}
Note that condition \eqref{eq:fermion-permutation-a} uniquely specifies the action of $U$ up to a phase, while \eqref{eq:fermion-permutation-b} fixes the phase. In particular, we can calculate the action of $U$ on a general Fock state by writing \begin{equation}
	\ket{x_0, \dots, x_{N-1}}_f = (a_0^\dagger)^{x_0} \cdots (a_{N-1}^\dagger)^{x_{N-1}} \ket{0 \dots 0}_f
\end{equation}
and letting $U_\sigma$ conjugate each creation operator. 
\end{definition}

Central to our qubit circuits will be controlled-$Z$ gates. We let $\CZ_{a,b}$ denote such a gate acting on qubits $a$ and $b$. The effect of this is a phase $(-1)^{x_a x_b}$ on the qubit state $\ket{x_a x_b}$. 
When multiple $\CZ$ gates share a common qubit, we call this a $\CZ$-fanout, as shown in Fig.~\ref{fig:cz_cnot_fanout}(a).

$\CZ$ gates are related to $\CNOT$ gates by conjugating the target qubits with Hadamard gates, as depicted in Fig.~\ref{fig:cz_cnot_fanout}(b). 
Another related gate is the parity gate, also known as the $\MOD_2$ gate, where several $\CNOT$ gates share the same target qubit. The $\MOD_2$ gate is related to the $\CNOT$-fanout by conjugating all qubits with Hadamard gates~\cite{moore1999fanoutparity}.

Another relevant unitary we will use is the parity transform $P$ which acts on computational basis states as
\begin{equation}
	P \ket{x_0, x_1, \dots, x_{N-1}} = \ket{p_0, p_1, \dots, p_{N-1}}, \text{ where } 
	p_k = \bigoplus_{i=0}^k x_k.
	\label{eq:parity}
\end{equation}
To avoid confusion between the parity \emph{gate} and the parity \emph{transform}, we will henceforth refer to the former as the $\MOD_2$ gate.
It is well-known that all four of $\CZ$-fanout, $\CNOT$-fanout, $\MOD_2$, and $P$ can be implemented in a circuit with depth $O(\log n)$~\cite{moore1999fanoutparity, fanout-lower-bound, Remaud_2025}. We give another derivation of a log-depth circuit for $P$ from a ternary tree perspective in Section~\ref{sec:ternary-transforms}.

\begin{figure}
\centering
\includegraphics[width=\linewidth]{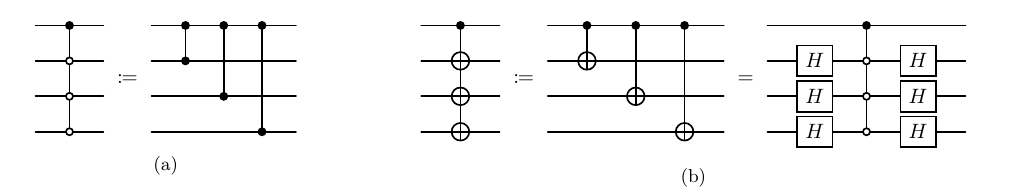}
\caption{ (a) The $\CZ$-fanout gate (left) is defined as a circuit consisting of one $\CZ$ gate between the control qubit (filled circle) and each target qubit (open circles). 
(b) The ordinary ($\CNOT$) fanout is related to the $\CZ$-fanout through conjugating each target qubit by Hadamards.
}\label{fig:cz_cnot_fanout}
\end{figure}

\section{Permutations in Jordan-Wigner}
\label{sec:jw-circuit}

In this section, we give an explicit circuit construction for any fermionic permutation under the Jordan-Wigner mapping. We begin by building circuits for transpositions that swap two fermionic modes. A general permutation can then be implemented as a product of swaps.

\begin{definition}[Fermionic swap]
The \emph{fermionic swap} operator $\fSWAP_{jk}$ is the fermion permutation operator in Definition \ref{def:fermion-permutation} that implements the transposition swapping modes $j$ and $k$.
\end{definition}

In the Fock basis, the fermionic swap operator acts on adjacent modes 
by swapping the occupation numbers and introducing a $-1$ phase if both are occupied \cite{bravyi_kitaev_2002}:
\begin{equation}
    \fSWAP_{i, i + 1} \ket{x_1, x_2, \dots, x_i, x_{i + 1}, \dots, x_{N}}_f = (-1)^{x_i x_{i + 1}}\ket{x_1, x_2, \dots, x_{i + 1}, x_i, \dots, x_{N}}_f.
\end{equation}
From this, the action of $\fSWAP_{ij}$ between modes with $i < j$ %
may be derived:
\begin{equation}\label{eq:fswap_part1}
    \fSWAP_{ij} \ket{x_0, \dots, x_i, \dots, x_j, \dots x_{N - 1}}_f = (-1)^p\ket{x_0, \dots, x_j, \dots, x_i, \dots, x_{N - 1}}_f,
\end{equation}
with
\begin{equation}\label{eq:fswap_part2}
    p = x_i \sum_{k = i + 1}^{j-1} x_k + x_j \sum_{k = i}^{j - 1} x_k.
\end{equation}
This equation also holds for the corresponding qubit states in the Jordan-Wigner encoding, which translates into two $\CZ$-fanout gates followed by a qubit swap, as shown in Fig.~\ref{fig:fswap-circuit}. 

\begin{figure}
    \centering
    \includegraphics[width=\textwidth]{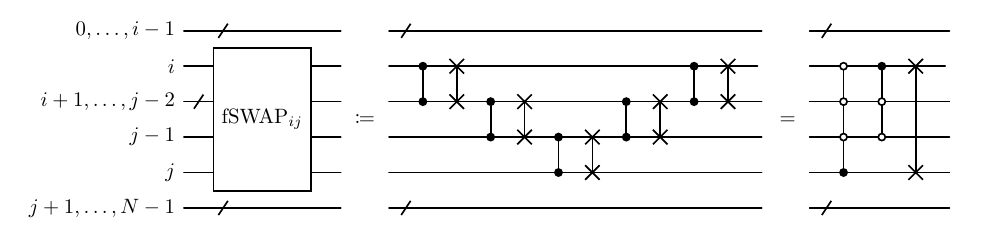}
    \caption{\label{fig:fswap-circuit}The circuit for performing the swap of modes $i$ and $j$ in the Jordan-Wigner encoding. The two $\CZ$-fanouts in the final equality are generated by commuting each $\CZ$ gate to the front of the $\SWAP$ circuit, and then collapsing the $\SWAP$ operators into a single $\SWAP$ of modes $i$ and $j$.}
\end{figure}

\begin{theorem}
\label{thm:jw-circuit}
Given any permutation $\sigma \colon [N] \to [N]$, there exists a circuit of depth $O(\log^2 N)$ that implements $\Phi_{\JW}(U_\sigma)$.
\end{theorem}

At a high level, we proceed in two stages. First, we decompose an arbitrary permutation into a product of $\ceil{\log_2 N}$ layers of disjoint \emph{staircase} permutations, which we define below. Then, we show how to implement each staircase permutation in depth $O(\log N)$, yielding a depth-$O(\log^2 N)$ routing circuit for any permutation.

We remark that these two stages resemble the construction in Ref.~\cite{maskara2025}, where an arbitrary permutation is decomposed into $O(\log N)$ layers of special efficient permutations. However, we differ in the choice of these special efficient permutations: we use ones that we call staircase permutations, while Ref.~\cite{maskara2025} uses others that they call interleaves. While some permutations are both interleaves and staircase permutations, neither is a subset of the other. Every staircase permutation can be implemented as a product of two interleaves \cite[Sec.~SI]{constantinides_optimal_2024}, and it remains open whether every interleave can be implemented via a constant number of staircase permutations.

\begin{lemma}
\label{lem:staircase}
Let a \emph{staircase} permutation be any permutation that can be written as a sequence of disjoint transpositions, $(m_1, n_1), \dots, (m_{k}, n_{k})$, such that the sequences $m_1, m_2, \dots, m_{k}$ and $n_1, n_2, \dots n_{k}$ are strictly increasing, and $m_k < n_1$. We call the interval $[m_1, n_k]$ its \emph{range}. Then every permutation $\sigma\colon [N] \to [N]$ can be written as a composition of $\lceil\log_2 N\rceil$ layers of staircase permutations, where each layer consists of a product of staircase permutations with disjoint ranges.
\end{lemma}

\begin{proof}
In the first layer, we partition $[N]$ into two halves: $L = \{1, 2, \dots, \floor{N/2}\}$ and $R = \{\floor{N/2}+1, \dots, N\}$. Let $x_1 < x_2 < \dots < x_k$ be all the elements of $L$ such that $\sigma(x_i) \in R$ for every $x_i$. Then there must be exactly $k$ elements of $R$, which we label $y_1 < y_2 < \dots < y_k$, with $\sigma(y_j) \in L$ for every $y_j$. We apply the staircase permutation $\sigma_1 = (x_1, y_1) \cdots (x_k, y_k)$, as depicted in Fig.~\ref{fig:staircase_perm_breakdown}. After applying $\sigma_1$, the left half contains exactly the elements mapping to $L$ and the right half contains exactly the elements mapping to $R$ (i.e.\ $(\sigma \circ \sigma_1^{-1})(L) = L$ and $(\sigma \circ \sigma_1^{-1})(R) = R$).

In subsequent layers, we recursively apply this construction within $\sigma_1(L)$ and $\sigma_1(R)$. Since the halves are disjoint, the two resulting staircase permutations are disjoint. As each layer operates on subsets of half the size of the prior layer, there are $\ceil{\log_2 N}$ layers.
\end{proof}

\begin{figure}[t]
    \centering
    \includegraphics[width=0.8\textwidth]{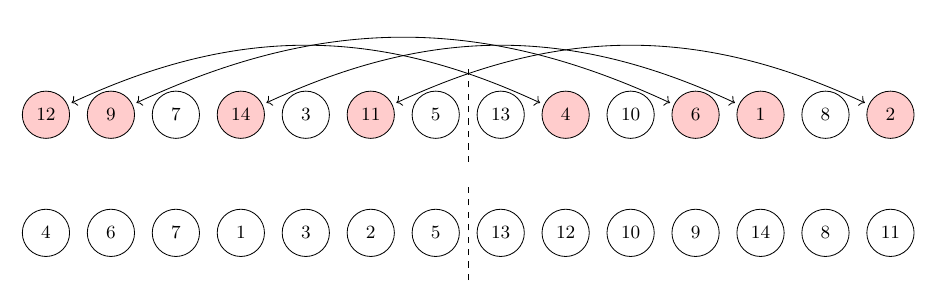}
    \caption{\label{fig:staircase_perm_breakdown}In this example, each element (indicated by a circle) is labeled with its target position under a permutation $\sigma$. The first step of implementing this permutation using staircase permutations is shown. Red marked elements have destinations across the dashed midpoint of the array, and are arranged in pairs and swapped in one staircase permutation step. Following this staircase step, the algorithm is run recursively on each side of the midpoint, treating the left and right halves separately.}
\end{figure}

Next, we introduce a technical circuit construction that is equivalent to \cite[Lemma 3]{maskara2025} but that we rederive here to demonstrate ancilla-free techniques.

\begin{lemma}
\label{lem:contiguous-range-cz}
Let a \emph{contiguous-range $\CZ$-fanout circuit} be a circuit consisting of $\CZ$-fanout gates with the following two properties:
\begin{enumerate}
\item Each fanout has a distinct control qubit on which no other gates act, and
\item Each fanout has targets spanning a contiguous range of qubits.
\end{enumerate} 
Then any contiguous-range $\CZ$-fanout circuit with support on $n$ qubits can be implemented in $O(\log n)$ depth.
\end{lemma}

\begin{proof}
As shown in Fig.~\ref{fig:parity-fanout}, a contiguous-range $\CZ$-fanout whose targets are conjugated with the parity ($P$) transformation turns into at most two $\CZ$ gates. Specifically, if the $\CZ$-fanout originally has control $0$ and targets $\{i, i + 1, \dots, j\}$, it becomes the gates $\CZ_{0, j}$ and $\CZ_{0, i-1}$ under conjugation by $P$, unless $i$ is the first qubit acted on by $P$, in which case it becomes the single $\CZ_{0, j}$ gate. Using this insight, we begin by conjugating by $P$ the shared target qubits of all $k$ 
contiguous-range $\CZ$-fanouts, as shown in the example of Fig.~\ref{fig:circuit-ranges}. As the control lines of each $\CZ$-fanout are distinct, we can separate the %
resulting $\CZ$ gates into two sets, marked blue and red, so that within each set, each $\CZ$ acts on a distinct qubit from the original set of controls. We depict this in Fig.~\ref{fig:circuit_ranges_grouped}. Finally, using the property that all $\CZ$ gates commute, we observe that the blue $\CZ$ gates and red $\CZ$ gates may be collapsed into two layers of at most $k$ $\CZ$-fanouts, where in each layer each fanout acts on disjoint qubits. Since a $\CZ$-fanout can be done in depth $O(\log n)$, this completes the proof. 
\end{proof}

\begin{figure}
    \centering
    \includegraphics[]{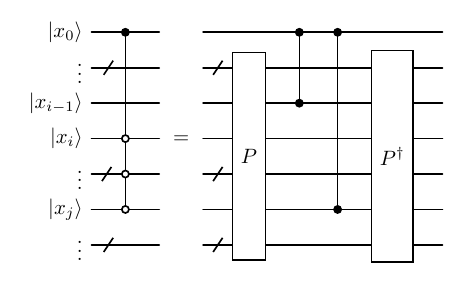}
    \caption{\label{fig:parity-fanout} The effect of conjugating the targets of a $\CZ$-fanout with the parity ($P$) gate, as defined in Eq.~\eqref{eq:parity}. The fanout, which has control $0$ and targets $\{i, \dots, j\}$, becomes two $\CZ$ gates, one $\CZ_{0, i - 1}$ and another $\CZ_{0, j}$. If the first qubit in $P$ is $i$, then the $\CZ_{0, i - 1}$ gate is not present.}
\end{figure}

\begin{figure}[h]
    \centering
    \includegraphics[width=\textwidth]{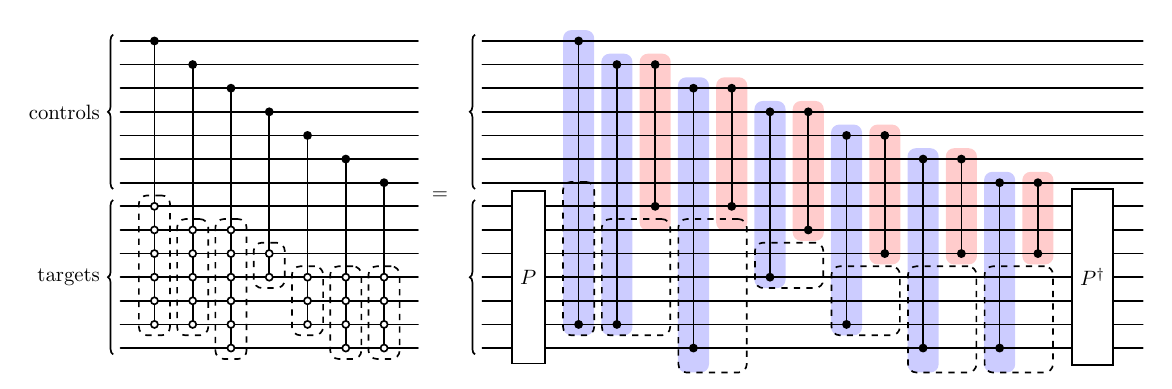}
    \caption{\label{fig:circuit-ranges}An example of a contiguous-range $\CZ$-fanout circuit. When transformed into the parity basis, the $\CZ$-fanouts turn into a sequence of individual $\CZ$ gates.}
\end{figure}

\begin{figure}[h]
    \centering
    \includegraphics[width=\textwidth]{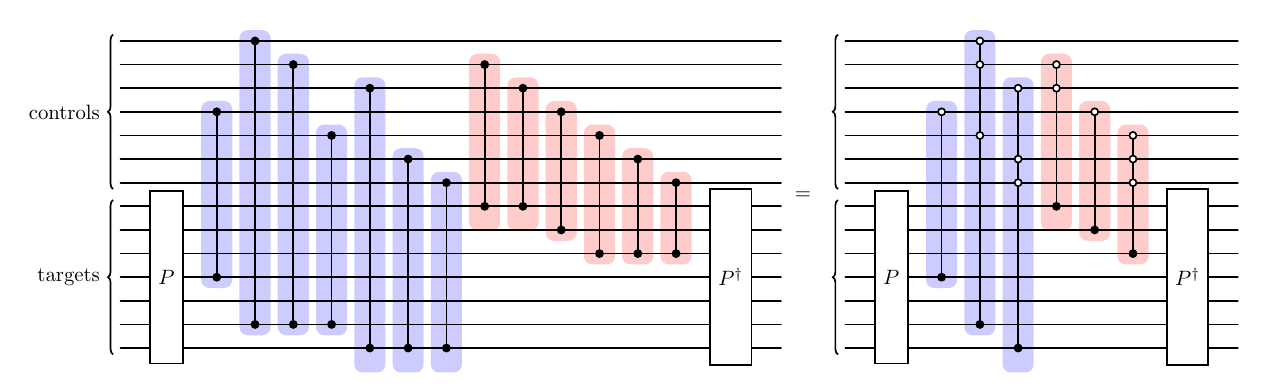}
    \caption{\label{fig:circuit_ranges_grouped} Continuing the example from Fig.~\ref{fig:circuit-ranges}, the $\CZ$ gates can be grouped into two sets (highlighted blue and red) of $\CZ$-fanouts. Within each set, the $\CZ$-fanouts act on (i.e., have as their targets) disjoint sets of qubits.}
\end{figure}

We are now ready to prove Theorem~\ref{thm:jw-circuit}.

\begin{proof}[Proof of Theorem~\ref{thm:jw-circuit}]
We begin by decomposing $\sigma$ into $O(\log N)$ layers of disjoint staircase permutations using Lemma~\ref{lem:staircase}. Since, in each layer, each staircase permutation has a disjoint range of elements,
by promoting each of the transpositions in the $\lceil \log_2 N \rceil$ layers of staircase permutations into fermionic swaps in the Jordan-Wigner encoding, one obtains a fermion routing circuit similarly made of $\lceil \log_2 N \rceil$ layers of fermionic staircase permutations acting on disjoint sets of qubits. This is seen by the property that each $\fSWAP_{ij}$ acts only on qubits $\{i, i + 1, \dots, j\}$, so each fermionic staircase permutation acts on disjoint sets of qubits in each layer, and thus the problem of fermion routing is reduced to implementing arbitrary fermionic staircase permutations, which we now show can be done in $O(\log N)$ depth.

To implement a general fermionic staircase permutation $(m_1,n_1)\dots(m_{k}, n_{k})$, we begin with a naive circuit implementing each transposition as two $\CZ$-fanouts followed by a qubit $\SWAP$, as shown in an example in Fig.~\ref{fig:starting-circuit}. We follow this example throughout the rest of our proof, though our discussion is general to any staircase permutation.
\begin{figure}[H]
\centering
\includegraphics[]{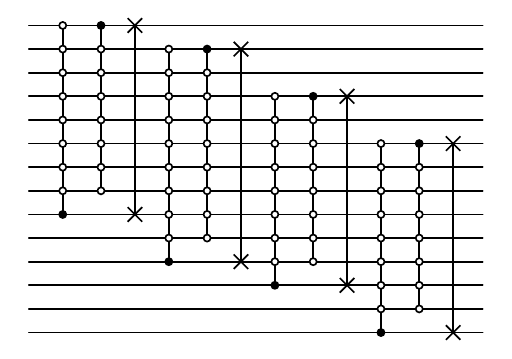}
\caption{\label{fig:starting-circuit}Example of a staircase of $\fSWAP$s expanded in terms of $\CZ$-fanouts and $\SWAP$ gates (see Fig.\ \ref{fig:fswap-circuit}).}
\end{figure}
\noindent Note that all $\CZ$ gates commute with each other, while commuting a $\SWAP$ with a $\CZ$ simply changes which qubits the $\CZ$ acts on, transporting the target qubit of a $\CZ$ from one end of the $\SWAP$ to the other. We commute all $\SWAP$ gates to the right of the naive circuit, shown for our example in Fig.~\ref{fig:commuted-circuit}. This updates the $\CZ$-fanouts by reassigning their targets, while the staircase property of the permutation ensures that the controls of each fanout do not change. Furthermore, we commute the $\CZ$-fanouts controlled by a qubit in $\{n_i\}_i$ to the left and group them as circuit $(1)$, and then group the fanouts with a control qubit in $\{m_i\}_i$ as circuit $(2)$. These are depicted for the example in dashed red boxes in Fig.~\ref{fig:commuted-circuit}. To elucidate a pattern generated circuit, we mark those qubits in the set $\{m_i\}_i$ blue, and those in the set $\{n_i\}_i$ red. The remaining qubits are unmarked and labeled in order starting from $1$.

\begin{figure}[H]
\centering
\includegraphics[]{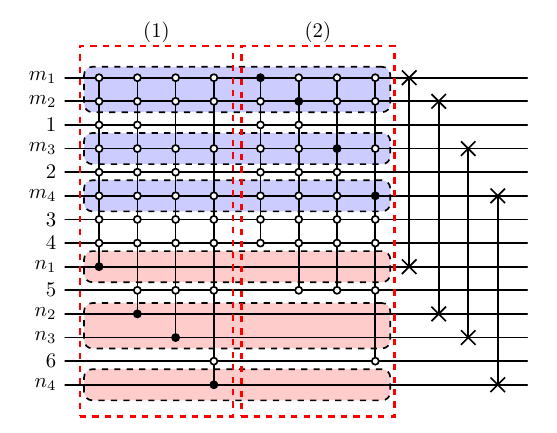}
\caption{\label{fig:commuted-circuit} The example of a staircase of $\fSWAP$s continued from Fig.~\ref{fig:starting-circuit}. The properties of $\SWAP$ gates and commutation of $\CZ$ gates are used to bring each $\CZ$-fanout to the left of the circuit. We commute the $\CZ$-fanouts into groups $(1)$ and $(2)$. In group $(1)$, each fanout has a control qubit in $\{n_i\}_i$ 
and in group $(2)$ each fanout has a control qubit in $\{m_i\}_i$.
Qubits $\{m_i\}_i$ and $\{n_i\}_i$ are labeled and marked blue and red, respectively. The rest of the qubits are labeled starting from $1$ in order, and unmarked.}
\end{figure}

We make some observations on the structure of the fanouts in the commuted circuit, such as in the example in Fig.~\ref{fig:commuted-circuit}, starting with circuit $(1)$. Specifically, the first fanout of each $\fSWAP_{m_i, n_i}$ operation in the naive circuit (e.g.~Fig.~\ref{fig:starting-circuit}) has control $n_i$, targets $\{m_i, m_i + 2, \dots, n_i - 1\}$, and is commuted through the set of $\SWAP$ gates $\{\SWAP_{m_1, n_1}, \SWAP_{m_2, n_2}, \dots, \SWAP_{m_{i - 1}, n_{i - 1}}\}$. This has the effect of reassigning a subset of its targets $\{n_{1}, n_2, \dots, n_{i - 1}\}$ to be those qubits $\{m_1, m_2, \dots, m_{i-1}\}$. Thus, its targets will always be given by the set of qubits $\{m_j\}_j$, marked in blue, and the unmarked qubits between $m_i$ and $n_i$, $\{m_i + 1, m_i + 2, \dots, n_i - 1\} \setminus (\{m_j\}_j \cup \{n_j\}_j)$, as is seen in Fig.~\ref{fig:commuted-circuit}. 
Similarly, in circuit $(2)$, the fanout with control $m_i$ originates from the second fanout of the $\fSWAP_{m_i, n_i}$ operation in the naive circuit (e.g.~Fig.~\ref{fig:starting-circuit}) with control $m_i$ and targets $\{m_i + 1, m_i + 2, \dots, n_i - 1\}$. It is commuted through the set of $\SWAP$ gates $\{\SWAP_{m_1, n_1}, \SWAP_{m_2, n_2}, \dots, \SWAP_{m_{i - 1}, n_{i - 1}}\}$, reassigning the subset of its targets $\{n_1, n_2, \dots, n_{i  - 1}\}$ to the qubits $\{m_1, m_2, \dots, m_{i - 1}\}$. Thus, its targets in circuit (2) will be given by the set of blue qubits $\{m_j\}_j \setminus \{m_i\}$ and every unmarked qubit between $m_i$ and $n_i$, $\{m_i + 1, m_i + 2, \dots, n_i - 1\} \setminus (\{m_j\}_j \cup \{n_j\}_j)$ (e.g.~Fig.~\ref{fig:commuted-circuit}).%

\begin{figure}[h]
\centering
\includegraphics[]{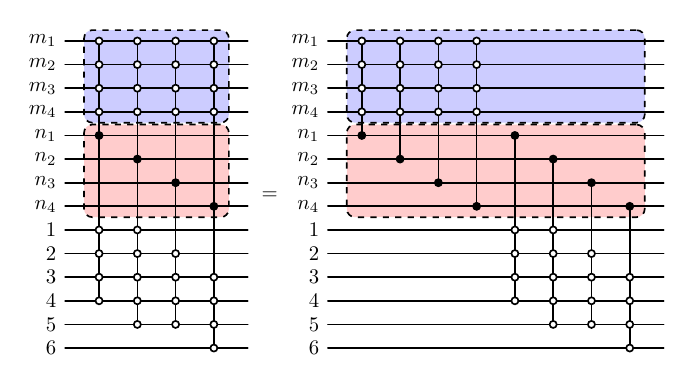}
\caption{Circuit (1) from Fig.~\ref{fig:commuted-circuit} after reordering the qubits to place the %
qubits in $\{m_i\}_i$ and $\{n_i\}_i$ at the top of the circuit.\label{fig:upgates}}
\end{figure}

We next simplify each of the circuits $(1)$ and $(2)$, as shown in the examples of Figs.~\ref{fig:upgates} and \ref{fig:circuit2_simplified}(a), respectively. In both circuits, we first reorder the qubits such that the blue qubits $\{m_i\}_i$ come first, then the red $\{n_i\}_i$, then the unmarked qubits. In circuit $(1)$ (e.g.~Fig.~\ref{fig:upgates}), we split each fanout gate into two fanout gates, one with control $n_i$ and blue targets $\{m_j\}_j$, and the other with control $n_i$ and unmarked targets qubits $\{m_i + 1, m_i + 2, \dots, n_i - 1\} \setminus (\{m_j\}_j \cup \{n_j\}_j)$. Lastly, in circuit $(2)$ (e.g.~Fig.~\ref{fig:circuit2_simplified}(a)), we cancel the targets of each fanout with control $m_i$ on the qubits $\{m_j\}_j \setminus \{m_i\}$ using the circuit identity shown in Fig.~\ref{fig:circuit2_simplified}(b). This leaves only the unmarked qubit targets $\{m_i + 1, m_i + 2, \dots, n_i - 1\} \setminus (\{m_j\}_j \cup \{n_j\}_j)$. Notice that for any $i$, the sets of blue qubits $\{m_j\}_j$, and unmarked qubits between $m_i$ and $n_i$, $\{m_i + 1, m_i + 2, \dots n_i - 1\} \setminus (\{m_j\}_j \cup \{n_j\}_j)$, now form contiguous ranges in the new ordering. Thus, each of the final circuits (1) and (2) (as in Fig.~\ref{fig:upgates} and Fig.~\ref{fig:circuit2_simplified}(a)) are composed of two and one contiguous-range $\CZ$-fanout circuits respectively, which, as proven in Lemma~\ref{lem:contiguous-range-cz}, have depth $O(\log N)$. 
\begin{figure}
\centering
\includegraphics[width=.9\linewidth]{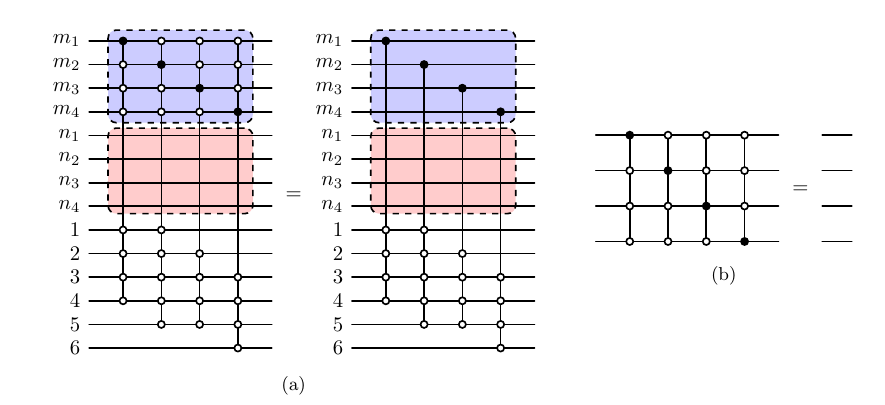}
\caption{(a) Circuit (2) of Fig.~\ref{fig:commuted-circuit} following the same reordering of qubits as in Fig.~\ref{fig:upgates}. The identity in (b) is used to cancel each fanout target on the qubits $\{m_i\}_i$. (b) A sequence of $k$ $\CZ$-fanouts, in which each fanout has a unique control line and $k-1$ other shared targets reduces to the identity, as a $\CZ_{ij}$ gate is generated for every pair of qubits $i \neq j$.}\label{fig:circuit2_simplified}
\end{figure}
\end{proof}

\section{Transformations between Encodings}
\label{sec:ternary-transforms}

In this section, we extend the regime of fermion-to-qubit mappings for which the fermion permutation circuit of Ref.~\cite{maskara2025} and Theorem~\ref{thm:jw-circuit} is applicable. We show that a broad range of mappings, called product-preserving ternary tree mappings, can be transformed into the Jordan-Wigner mapping in depth $O(\log^2 N)$. Consequently, routing in those mappings can also be implemented in depth $O(\log^2 N)$ by first transforming to Jordan-Wigner, applying our routing circuit, and then transforming back to the original map.

Let $\phi_1$ and $\phi_2$ be two fermion-to-qubit mappings over the same fermion and qubit spaces. Since we assume $\dim \mathcal{H}_f = \dim \mathcal{H}_q = 2^N$, they are unitary maps. Then there exists a unitary $U \in \mathcal{B}(\mathcal{H}_q)$ transforming between the two mappings such that $\phi_2 = U\phi_1$, namely $U = \phi_2 \phi_1^{-1}$. This transformation acts on operators via conjugation:
\begin{equation}
\Phi_2(\mathcal{O}_f) = U \Phi_1(\mathcal{O}_f) U^\dagger.
\end{equation}
We construct efficient circuits implementing $U$ when $\phi_1$ and $\phi_2$ lie in a class of fermion-to-qubit mappings known as product-preserving ternary tree mappings. 

Ternary tree mappings, also known as qubit trees in other contexts~\cite{Vlasov_2022}, were introduced in Ref.~\cite{Jiang_2020} to give a mapping that provably minimizes the Pauli weight of $\Phi(a_k)$ averaged over all $k$. However, the initial construction only focused on the mapping of operators and not of states. Subsequent works \cite{Bonsai_2023,chiew2024} extended the construction to study properties of the mapping of states. 
In this paper, we use a formulation of ternary trees adapted from Ref.~\cite{yu2025}.
See any standard classical algorithms text (e.g.~Ref.~\cite{clrs}) for the basics of trees and tree traversal algorithms.

\begin{definition}
A \emph{full ordered ternary tree} is a rooted tree where each node has exactly $0$ or $3$ children. A node with $0$ children is a \emph{leaf} and a node with $3$ children is a \emph{parent}. The $3$ children of a given parent are ordered and named the left, middle, and right child.
\end{definition}

Note that a full ternary tree with $N$ parent nodes always has $2N+1$ leaves.

\begin{definition}[Ternary tree mapping]
\label{def:ternary-tree-mapping}
A ternary tree mapping consists of the following information:
\begin{enumerate}
\item A full ordered ternary tree, which we call the \emph{shape} of the ternary tree mapping.
\item A labeling of the $N$ parents by $1, \ldots, N$. Each labeled parent corresponds to a qubit.
\item A labeling of the $2N+1$ leaves by $\gamma_1, \ldots, \gamma_{2N+1}$. These label the Majorana operators.
\end{enumerate}
\end{definition}

Reference~\cite{Jiang_2020} showed that a ternary tree mapping with $N$ parents gives a fermion-to-qubit mapping from $N$ fermionic modes to $N$ qubits by associating each leaf $\gamma_j$ with a Pauli string as follows. Traverse down the tree from the root to $\gamma_j$. Along the path, taking the left, middle, or right child of qubit $k$ appends to the Pauli string a Pauli $X_k$, $Y_k$, or $Z_k$, respectively. The mapping $\Phi$ of fermionic operators is then defined by $\Phi(a_k) =\gamma_{2k-1} + i\gamma_{2k}$. The last leaf $\gamma_{2N+1}$ is redundant and discarded. Throughout this work, we do not explicitly keep track of the signs of the Majorana Paulis---as detailed in Ref~\cite{chiew2024}, Pauli sign changes give an equivalent encoding up to single-qubit Paulis which take depth $1$, and as such do not affect asymptotic scaling.

\begin{figure}[t]
\includegraphics[width=0.8\linewidth]{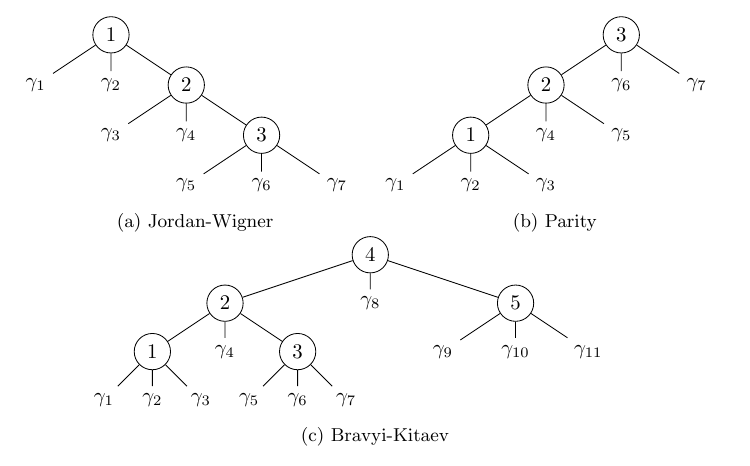}
\caption{Many common fermion-to-qubit mappings, including (a) Jordan-Wigner, (b) Parity, and (c) Bravyi-Kitaev, can be expressed as ternary trees. All three examples shown are binary-shaped ternary tree mappings.}
\label{fig:jw-tree}
\end{figure}

\begin{definition}[Binary subtree]
Given a ternary tree, the binary subtree is the subgraph induced by the root and recursively all left and right children of nodes in the subtree. We call a tree \emph{binary shaped} if all $N$ qubit nodes lie in the binary subtree.
\end{definition}

The Jordan-Wigner, Bravyi-Kitaev, and Parity (defined as $\phi_P \ket{x_0, \dots, x_{N-1}}_f = \ket{p_0, \dots, p_{N-1}}$ with $p_k$ given by Eq.~\eqref{eq:parity}) mappings can all be represented as binary-shaped ternary tree mappings~\cite{Bonsai_2023}, as shown in Fig.~\ref{fig:jw-tree}. Furthermore, in all three examples, the qubits are always labeled via an inorder traversal, and the leaves are labeled in order from left to right. (Recall that an inorder traversal is an ordering of nodes where we first recursively perform an inorder traversal of the left subtree, then visit the root, and then recursively perform an inorder traversal of the right subtree.)

To design circuits that transform between mappings, we first recall the effects of certain gates on ternary tree mappings.

\begin{lemma}[\cite{vlasov2024trees, yu2025}]
\label{lem:tree-ops}
Let $\phi$ be a ternary tree mapping. Then for each $U$ below, $U\phi$ is another ternary tree mapping related to $\phi$ as described.
\begin{enumerate}[(a)]
\item $U = \SWAP_{jk}$: Swap the qubit labels $j$ and $k$ in the tree.
\item $U = \CNOT_{jk}$, where node $j$ is the left child of node $k$: Perform the %
right tree rotation of qubit node $j$ about node $k$ in the binary subtree, preserving the parents of middle children. The right tree rotation is defined as shown in Fig.~\ref{fig:tree-rotation}.
\item $U = \CNOT_{jk}$, where $k$ is the right child of node $j$: Perform the left tree rotation of qubit node $k$ about node $j$ in the binary subtree. This is the inverse of (b) and is defined as shown in Fig.~\ref{fig:tree-rotation}.
\item $U = S_k$, where $S = \pmat{1 &  0 \\ 0 & i}$ is the phase gate: Swap the left and middle children (carrying along their subtrees) of the node labeled $k$. %
\end{enumerate}
\end{lemma}

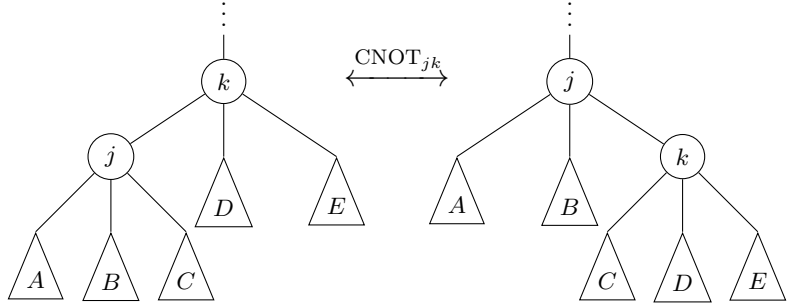
\begin{figure}[h]
\centering
	\begin{tikzpicture}
	[
	style = {level distance = 1cm}, 
	level 2/.style = {sibling distance = 1.5cm},
	level 3/.style = {sibling distance = 1cm},
	internal/.style={draw, circle},
	subtree/.style={draw, isosceles triangle,shape border rotate=90,anchor=apex,minimum size=5mm,text height=0.5ex, text depth=0.25ex},
	baseline=(current bounding box.center)
	]
	\node [xshift=-2.3cm] {$\vdots$} child {
	node [internal] {$k$}
		child {node [internal] {$j$}
			child [child anchor=apex] {node [subtree] {$A$}}
			child [child anchor=apex] {node [subtree] {$B$}}
			child [child anchor=apex] {node [subtree] {$C$}}
		} 
		child {node [subtree] {$D$}}
		child [child anchor=apex] {node [subtree] {$E$}}
	};
	\node [yshift=-0.8cm] {\large $\xleftrightarrow{\CNOT_{jk}}$};
	\node [xshift=2.3cm] {$\vdots$} child {
	node [internal] {$j$}
		child [child anchor = apex] {node [subtree] {$A$}}
		child [child anchor = apex] {node [subtree] {$B$}}
		child {node [internal] {$k$}
			child [child anchor = apex] {node [subtree] {$C$}}
			child [child anchor = apex] {node [subtree] {$D$}}
			child [child anchor = apex] {node [subtree] {$E$}}
		} 
	};
	\end{tikzpicture}
\caption{Tree rotation corresponding to the adjoint action of the CNOT$_{jk}$ transformation. The triangles can be leaves or contain further subtrees.}
\label{fig:tree-rotation}
\end{figure}

\begin{theorem}[Folklore]
There exists a $\CNOT$ circuit of depth $O(\log N)$ that transforms between the Jordan-Wigner, Bravyi-Kitaev, and Parity mappings.
\label{thm:bk-jw}
\end{theorem}

\begin{proof}
First, we demonstrate a transformation between Bravyi-Kitaev and Jordan-Wigner. We construct $O(\log N)$ layers of disjoint $\CNOT$ gates to transform the Bravyi-Kitaev tree to the Jordan-Wigner tree.
Call the \emph{right spine} of a tree the root and all nodes which are right descendants of nodes in the spine. In each layer, iterate through every node $k$ on the right spine with a left child $j$ and perform $\CNOT_{jk}$ to rotate node $j$ onto the right spine. For each node not already on the right spine, its distance to the right spine decreases by $1$ in each iteration. For Bravyi-Kitaev, the farthest starting node from the right spine has distance $O(\log N)$, so this procedure will terminate in $O(\log N)$ iterations. Since tree rotations leave invariant the inorder traversal, the resulting ternary tree also has qubits and leaves each labeled via an inorder traversal, thus exactly matching the Jordan-Wigner tree. An example of this protocol is illustrated in Fig.~\ref{fig:bk-jw}.

The transformation from Bravyi-Kitaev to Parity is simply the mirrored version of the above, with the left spine and left rotations instead. These transformations are unitary and invertible, so by composing Bravyi-Kitaev-to-Parity with the inverse of Bravyi-Kitaev-to-Jordan-Wigner, we obtain a log-depth circuit transforming between Jordan-Wigner and Parity.
\end{proof}

\begin{figure}[t]
\centering
\includegraphics{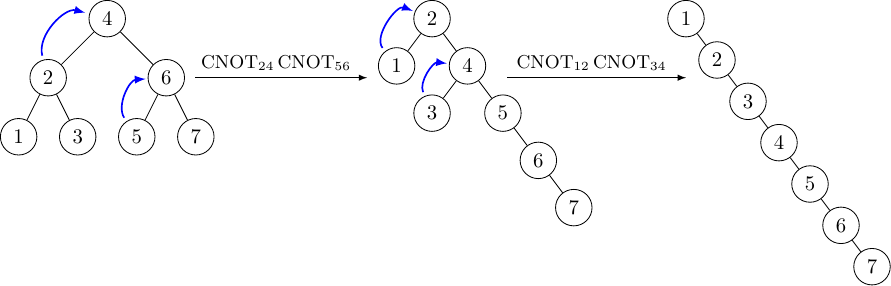}
\caption{For $N = 7$ qubits, a protocol transforming the Bravyi-Kitaev tree into the Jordan-Wigner tree using $\floor{\log_2 N}$ layers of $\CNOT$ gates. The Majorana leaves are not drawn but are always implicitly labeled in order.}
\label{fig:bk-jw}
\end{figure}

We remark that, while Theorem~\ref{thm:bk-jw} was formulated in the context of fermion-to-qubit mappings, the circuits can also transform qubit states in the same manner, outside the context of any fermions. In particular, the same circuit constructed in Theorem~\ref{thm:bk-jw} to transform Jordan-Wigner to Parity also gives a log-depth circuit implementing the parity transform [Eq.~\eqref{eq:parity}] on qubit states.
We note that this reproduces the same circuit as the well-known divide-and-conquer $\CNOT$ circuit for parity.

\begin{algorithm}[H]
    \caption{Circuit transforming any binary-shaped ternary tree into the Jordan-Wigner shape}
    \label{alg:jw}
    
    \begin{algorithmic}[1]
	\While {tree is not Jordan-Wigner shape}
		\ForAll {nodes $k$ on the right spine}
			\If {$k$ has a non-leaf left child $j$}
				\State Apply $\CNOT_{jk}$
			\EndIf
		\EndFor
	\EndWhile
    \end{algorithmic}
\end{algorithm}

We also remark that the technique of rotating nodes onto the right spine is not restricted to starting from Bravyi-Kitaev. Rather, this technique can be applied to transform any binary-shaped ternary tree to the Jordan-Wigner shape and is detailed in Algorithm~\ref{alg:jw}.

\begin{lemma}
Let $T$ be a binary-shaped ternary tree mapping and let $d$ be the furthest distance from any node to the right spine. There exists a $\CNOT$ circuit of depth $d$ that transforms between $T$ and the Jordan-Wigner tree shape.
\label{lem:log-jw}
\end{lemma}

\begin{proof}
The circuit is constructed by Algorithm~\ref{alg:jw}. For each node not already on the right spine, its distance to the right spine decreases by $1$ in each iteration of the \textbf{while} loop. Therefore the \textbf{while} loop runs $d$ times. Within each iteration, all $\CNOT_{jk}$ gates can be applied in parallel as they act on disjoint qubits. Therefore the resulting circuit has depth $d$.
\end{proof}

\begin{theorem}
\label{thm:binary-shaped}
There exists a $\CNOT$ circuit of depth $O(\log N)$ transforming between any two binary-shaped ternary trees on $N$ qubits sharing the same ordering of Majorana leaf labeling.
\end{theorem}

\begin{proof}
For this proof, we only keep track of the qubit nodes; the leaves are always attached in order since tree rotations leave the inorder traversal invariant.

It suffices to show that there exists a circuit of depth $O(\log N)$ transforming any binary-shaped tree to the Jordan-Wigner tree shape. To do so, we first use parallel tree rotations to transform an arbitrary binary tree into a binary tree with depth $O(\log N)$. The depth upper bounds the maximum distance of any node from the right spine, so by Lemma~\ref{lem:log-jw} we can transform the log-depth tree into the Jordan-Wigner shape in depth $O(\log N)$.
By the inorder invariant of tree rotations, the Majorana leaves remain labeled in the same order. The only remaining degree of freedom is in the qubit labels, which can be fixed in depth two by decomposing the desired qubit permutation into two layers of $\SWAP$ gates~\cite{alon1994routing}.

To attain a log-depth binary-shaped tree, we introduce some notation to identify the imbalances in the tree. The \emph{weight} of a node $A$, denoted $\wt(A)$, is the number of leaves in the subtree rooted at $A$. The \emph{halving depth} of $A$, denoted $\hd(A)$, is the minimum number of levels below $A$ such that every descendant $B$ of $A$ in that level has at most half the weight, i.e., $\wt(B) \leq \tfrac{1}{2} \wt(A)$. 

There is clearly at most one node per level with weight $>\tfrac{1}{2} \wt(A)$---there is not enough weight in the whole subtree for more than one. In particular, some $B$ exists $\hd(A)-1$ steps below $A$ such that $\wt(B) > \tfrac{1}{2} \wt(A)$. By inclusion, all the ancestors of $B$ up to $A$ are also above the threshold. Inspired by this, we define the \emph{heavy path from $A$} to be the sequence of descendants (of length $\hd(A)$) that have weight $> \tfrac{1}{2} \wt(A)$.

The algorithm runs $\bigo{\log N}$ \emph{rounds}. In each round, we consider all nodes where the depth is a multiple of $3$, e.g., the root is depth $0$, its great-grandchildren are depth $3$, and so on. For each of these nodes $A$, if $\hd(A) \geq 4$ then the heavy path begins with nodes $A, A', A'', A'''$. The algorithm will shorten the heavy path using tree rotations, but the choice of rotation depends on which child $A'$ is to $A$, and similarly on $A''$ relative to $A'$. Table~\ref{tab:heavy-rotations} summarizes which rotation to use for each case.  
\begin{table}[b]
\begin{tabular}{|cc|cc|}
	\hline
	& & \multicolumn{2}{c|}{$A'$ to $A''$} \\
	& & left & right \\
	\hline
	\multirow{2}{*}{$A$ to $A'$} & left & $\CNOT_{A''A'}$ (Fig.~\ref{fig:tree-rotation}) & $\CNOT_{A''A} \CNOT_{A'A''}$ (Fig.~\ref{fig:zigzag-rotations}(a)) \\
	& right & $\CNOT_{AA''} \CNOT_{A''A'}$ (Fig.~\ref{fig:zigzag-rotations}(b)) & $\CNOT_{A'A''}$ (Fig.~\ref{fig:tree-rotation}) \\
	\hline
\end{tabular}
\caption{The choice of tree rotations in Algorithm~\ref{alg:binary-shaped}.}
\label{tab:heavy-rotations}
\end{table}
The key property is that in all cases, we reduce the length of three steps in a heavy path to two or less. %

\begin{figure}
	\includegraphics[width=\linewidth]{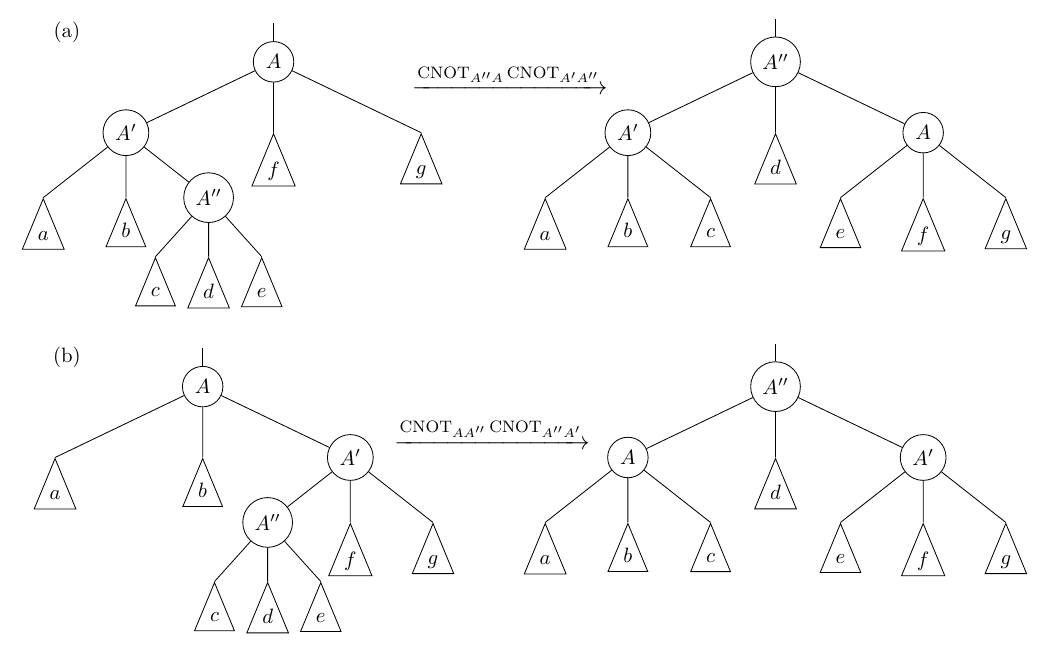}
	\caption{In Theorem~\ref{thm:binary-shaped}, depending on the shape of the heavy path, these rotations may be applied. In both cases, $A'''$ can be the root of either subtree $c$, $d$ or $e$.} 
	\label{fig:zigzag-rotations}
\end{figure}

Now consider how a round of the algorithm affects the halving depth $d = \hd(X)$ of an arbitrary node $X$. The heavy path has $d$ nodes, $X = X^{(1)}$ through $X^{(d)}$. The last $3$ ($X^{(d-2)}$, $X^{(d-1)}$, $X^{(d)}$) are too close to the end to trigger a tree rotation in the algorithm. Of the remaining nodes, at least $\lfloor \tfrac{d}{3} \rfloor-1$ appear on a level that is a multiple of three. %

The tree rotations usually change the root of the subtree; call the new root $Y$. The final element of the original heavy path, $X^{(d)}$, exists in the subtree and still has weight $\wt(X^{(d)}) > \tfrac{1}{2}\wt(Y)$. However, the algorithm has applied no tree rotation that could affect $X^{(d)}$, so its children have weight $\leq \tfrac{1}{2}\wt(Y)$. Hence, the halving depth of $Y$ is its distance from $X^{(d)}$. The distance between the subtree's root and $X^{(d)}$ decreases by $\lfloor \tfrac{d}{3} \rfloor-1 \geq \tfrac{d-5}{3}$ from the tree rotations, %
so 
\[
\hd(Y) \leq d - \frac{d-5}{3} \leq \frac{2 d + 5}{3}.
\]
We conclude from this inequality that the halving depth of any subtree decreases geometrically to a constant. 
The initial halving depth of any subtree is at most the (standard) depth of the tree, and thus bounded by $n$. It follows that after $\bigo{\log N}$ rounds, the halving depth of any subtree is $\bigo{1}$.  

Suppose the halving depth is bounded by a constant $H$ for all subtrees. Then the (standard) depth of the tree is at most $H(\log_2(n) + \bigo{1})$, since the weight is $2n+1$ at the root, halves every $H$ levels, and never goes lower than $3$. Since 
$H = \bigo{1}$,
we conclude that the standard depth is at most $\bigo{\log N}$. 
\end{proof}

\begin{algorithm}[H]
    \caption{Circuit transforming any binary-shaped ternary tree into the Jordan-Wigner shape}
    \label{alg:binary-shaped}
    
    \begin{algorithmic}[1]
	\Procedure{balance-binary}{}
	\While {some node has $\hd \ge 6$}
		\ForAll {nodes $A$ with depth divisible by $3$}
			\If {$\hd(A) \ge 4$}
				\State $A', A'', A''' \gets$ first three nodes in heavy path from $A$
				\State Apply $\CNOT$ rotations according to Table~\ref{tab:heavy-rotations}
			\EndIf
		\EndFor
	\EndWhile
	\EndProcedure

	\State Apply \Call{balance-binary}{} to the starting tree
	\State Apply Algorithm~\ref{alg:jw} on the resulting tree
    \end{algorithmic}
\end{algorithm}

\begin{definition}[Product preservation]
\label{def:product-preserving}
A fermion-to-qubit mapping $\phi$ is product-preserving if $\phi(\ket{\psi}_f)$ is a computational basis state in $\mathcal{H}_q$ for every Fock basis state $\ket{\psi}_f$.
\end{definition}

\begin{theorem}
\label{thm:ternary-jw}
There exists a circuit of depth $O(\log^2 N)$, consisting of $\CNOT$ and $S$ gates, transforming between any two product-preserving ternary tree mappings.
\end{theorem}

\begin{proof}
It suffices to give an algorithm transforming any product-preserving ternary tree mapping into the Jordan-Wigner map. At a high level, we proceed in two stages. In the first stage, we apply Algorithm~\ref{alg:ternary-jw} to map any ternary tree into the Jordan-Wigner shape. In the second stage, we use the theory of product-preserving ternary trees from Ref.~\cite{chiew2024} to argue that the resulting tree cannot be too far from the actual Jordan-Wigner tree.

To begin the implementation of Algorithm~\ref{alg:ternary-jw}, as illustrated in Fig.\ \ref{fig:ternary-jw}, we first apply Algorithm~\ref{alg:binary-shaped} to transform the binary subtree of the root, as well as the binary subtree of every middle child node, into one with no left children.
All of these subtrees are independent of each other, so all their transformations can be done in parallel in depth $O(\log N)$. 

Now we are left with a ternary tree where all left children are leaves. For every qubit node $k$ that has a middle child qubit node, apply the phase gate $S$ on qubit $k$. This swaps the left and middle children (Lemma~\ref{lem:tree-ops}) so that the left child has a qubit node and middle child a leaf. This makes the tree binary-shaped, so we can apply Algorithm~\ref{alg:binary-shaped} to convert it into the Jordan-Wigner tree shape, namely a single right spine. This completes the first stage.

For the second stage, we note that if $\ket{\psi}$ is a computational basis state, then so are $\CNOT_{jk} \ket{\psi}$ and $S_k \ket{\psi}$ for any $j$ and $k$. Since all gates used in Algorithm~\ref{alg:ternary-jw} maintain the product preservation property, our resulting tree is also product-preserving. In particular, the vacuum state $\ket{0, \dots, 0}_f$ maps to some computational basis state $\ket{q_1, \dots, q_n}$. We apply the Pauli gates $X_1^{q_1} \otimes \dots \otimes X_n^{q_n}$ so that the vacuum state now maps to the computational basis state $\ket{0^n}$. This has no effect on the tree shape, since conjugation of a Pauli string by a Pauli can only multiply the string by a phase factor.

We now have a ternary tree with the Jordan-Wigner shape that maps the vacuum state to $\ket{0^n}$. By \cite[Lemma 5.9]{chiew2024}, all ternary tree mappings with the Jordan-Wigner shape mapping the vacuum state to $\ket{0^n}$ must be equivalent to the exact Jordan-Wigner mapping up to a fermionic permutation, pair braiding (paired Majoranas $(\gamma_{2k-1}, \gamma_{2k})$ transforming to  $(\gamma_{2k}, -\gamma_{2k-1})$), and Pauli sign changes. The latter two may be fixed by constant-depth single-qubit gates, leaving the fermion permutation as the only degree of freedom.

This means that, given any two product-preserving ternary-tree mappings $\phi_1$ and $\phi_2$, there exist circuits ${C}_1$ and ${C}_2$ of depth $O(\log^2 N)$ such that ${C}_1 \phi_1 = \phi_{\JW} U_\sigma$ and ${C}_2 \phi_2 = \phi_{\JW} U_\tau$ for some permutations $\sigma, \tau \colon [N] \to [N]$. (Recall from Eq.~\eqref{eq:operator-transform} that $\phi_{\JW} U_f = \Phi_{\JW}(U_f) \phi_{\JW}$.) By Theorem~\ref{thm:jw-circuit}, there exists a circuit ${D}$ that implements $\Phi_{\JW}(U_{\tau\sigma^{-1}})$ in depth $O(\log^2 N)$. Then $C_2^{-1} D C_1$ is a circuit of depth $O(\log^2 N)$ transforming from $\phi_1$ to $\phi_2$.
\end{proof}

\begin{algorithm}[H]
    \caption{Circuit transforming any ternary-tree shape into a right spine (Jordan-Wigner shape)}
    \label{alg:ternary-jw}
    
    \begin{algorithmic}[1]
	\ForAll {$r \in \{\text{root node, every qubit node that is a middle child}\}$}
		\State Apply Algorithm~\ref{alg:binary-shaped} to the binary subtree rooted at $r$
	\EndFor
	\ForAll {$k \in \{\text{qubit nodes}\}$}
		\If {$k$'s middle child is not a leaf}
			\State Apply $S_k$
		\EndIf
	\EndFor
	\State Apply Algorithm~\ref{alg:binary-shaped} on the resulting binary-shaped tree
    \end{algorithmic}
\end{algorithm}

\begin{figure}[t]
\centering
\includegraphics{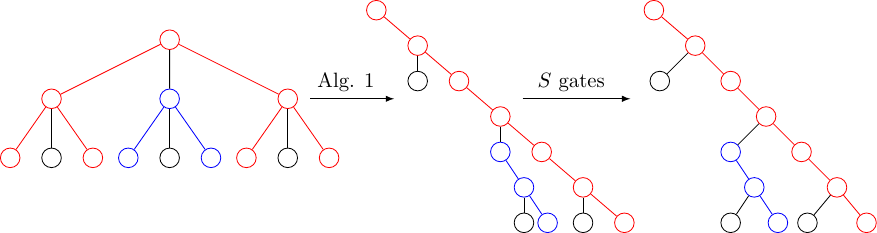}
\caption{In this example execution of Algorithm~\ref{alg:ternary-jw}, we first apply Algorithm~\ref{alg:binary-shaped} on the red binary subtree as well as the blue binary subtree. Then we apply $S$ gates on all middle nodes, depicted with a black edge, to transform them to left nodes. We would then apply Algorithm~\ref{alg:binary-shaped} again on the resulting binary tree. The qubit labelings and Majorana leaves are omitted.}
\label{fig:ternary-jw}
\end{figure}

Finally, we combine Theorems~\ref{thm:jw-circuit} and~\ref{thm:ternary-jw} to prove Theorem~\ref{thm:main}, which we now restate in technical language.

\newcounter{tempcounter}
\setcounter{tempcounter}{\value{theorem}}
\setcounter{theorem}{0}
\begin{theorem}
Given any product-preserving ternary tree mapping $\phi$ on $N$ fermions and any permutation $\sigma \colon [N] \to [N]$, there exists a quantum circuit that implements $\Phi(U_\sigma)$ in depth $O(\log^2 N)$.
\end{theorem}
\setcounter{theorem}{\value{tempcounter}}

\begin{proof}
By Theorem~\ref{thm:ternary-jw}, there exists a circuit $C$ of depth $O(\log^2 N)$ such that $\mathcal{C}\phi = \phi_{\JW}$. By Theorem~\ref{thm:jw-circuit}, there exists a circuit $D$ that implements $\Phi_{\JW}(U_{\sigma})$ in depth $O(\log^2 N)$. Then $C^{-1}DC$ implements $\Phi(U_\sigma)$ in depth $O(\log^2 N)$, as desired.
\end{proof}

\section{Applications}\label{sec:applications}

In this section, we discuss applications of our methods to time evolution under widely studied fermionic models, see also Ref.~\cite{maskara2025} for applications of log-depth routing. We summarize our discussions in Table~\ref{tab:comparison}, where we compare to the best previous methods that we are aware of.

Simulating fermionic Hamiltonians is one of the most promising applications of quantum computing, enabling the solution of classically challenging problems in quantum chemistry~\cite{McArdle2020}, materials science~\cite{babbush_low-depth_2018}, and high-energy physics~\cite{bauer_quantum_2023}. However, in order to simulate fermions with qubit quantum computers, the fermionic operators must be mapped to 
qubit operators. This leads to a multiplicative overhead in the circuit depth and gate count. %
In other words, this overhead is defined as the depth or gate count on a conventional qubit quantum computer (see e.g.~Refs.~\cite{stanisic2022observing,hemery_measuring_2024,nigmatullin2025experimental,evered_probing_2025} for recent experimental demonstrations) divided by the depth or gate count on a fermionic quantum computer that digitally manipulates fermionic particles directly (see e.g.~Refs.~\cite{yan2022twodimensional,gonzalez-cuadra2023fermionic,schuckert_fermion-qubit_2024,ott_error-corrected_2024}). 

The simplest mapping of fermions to qubits is the Jordan-Wigner encoding \cite{jordanwigner}. One way to understand this mapping is that it maps products of fermionic operators to products of qubit operators. Naively, this encoding has worst-case depth overhead linear in $N$ for $N$ fermionic sites because some fermionic Hamiltonian terms are mapped to Pauli strings with support on $\Theta(N)$ sites, even if in the fermionic terms only have support on $\Theta(1)$ sites. This is because the Jordan-Wigner encoding labels the fermionic modes along a fictitious one-dimensional chain. For example, the hopping term from one end of the chain to the other has weight $2$ in the fermionic formulation and weight $N$ after mapping to qubits. This is a highly geometrically non-local hopping term, but because of the effective one-dimensionality, even geometrically local hopping terms in low dimension can have high weight if they are not along the direction of the Jordan-Wigner chain. For example, the hopping term in two dimensions which is not along the chain has weight $O(\sqrt{N})$~\cite{hemery_measuring_2024}. While it is widely known that this leads to an at-worst linear depth overhead for fermion time evolution, the reason for this is often miscommunicated: it is not the high weight of the fermion terms which leads to the depth overhead, but the restriction in parallelization. The most straightforward approach for implementing time evolution under a Pauli-string Hamiltonian is done with a so-called Pauli gadget, i.e., the conjugation of a single-qubit rotation with an $N$-qubit CNOT ladder, which requires depth $N$. However, this CNOT ladder can actually be compressed to depth $O(\log N)$ without ancillas~\cite{Remaud_2025}. Therefore, a single fermionic Hamiltonian term can be implemented in depth $O(\log N)$ within the Jordan-Wigner encoding, and therefore the depth overhead resulting from the $O(N)$ operator weight is only $O(\log N)$. Instead, the reason why the Jordan-Wigner encoding has $O(N)$ 
depth overhead for the overall time-evolution task is the parallelization restriction: in the worst case, only one Hamiltonian term can be applied at a time, resulting in a depth scaling with the number of Hamiltonian terms, which is $\Omega(N)$ for a non-trivial Hamiltonian. 
By contrast, for fermionic quantum computers, there is no such restriction other than the connectivity graph of the Hamiltonian---in most physical Hamiltonians, the connectivity is such that $\Theta(N)$ terms can be implemented in parallel (this is the case in the 2D Fermi-Hubbard model~\cite{campbell2021early}), 
such that for example $\Theta(N)$ Hamiltonian terms can be implemented with depth $\Theta(1)$. 

The Bravyi-Kitaev encoding~\cite{bravyi_kitaev_2002} might at first seem to address this issue~\cite{Seeley_2012}, as it reduces the operator weight after the fermion-to-qubit mapping to only $O(\log N)$ sites. However, unfortunately, these savings for a single fermionic Hamiltonian term do not carry over to the task of implementing all terms: the Pauli strings to which products of fermion creation and annihilation operators are mapped all have one qubit in common, so they cannot be implemented in parallel \cite{Bringewatt2023}. Therefore, in practice, even though the weight reduction can lead to a gate count improvement, for most problems the Bravyi-Kitaev encoding does not show an improvement over the $O(N)$ 
depth overhead of the Jordan-Wigner encoding. Schemes have been proposed to avoid this restriction by using ancillas~\cite{bravyi_kitaev_2002,chien2020,Bringewatt2023,luo2025, schuckert_fermion-qubit_2024}, which in particular can achieve $O(1)$ overhead for geometrically local models~\cite{Verstraete2009,derby2021}. However, all of these approaches require $\Omega(N)$ ancillas, which might be a prohibitive overhead, especially for fault-tolerant architectures, where qubit count is often more expensive than depth due to the exponential error suppression~\cite{gidney2025how}. 

Instead, we here rely on the polylog routing schemes from Ref.~\cite{maskara2025} and this work to reduce the depth of fermion time evolution.

\subsection*{Reducing the fermion-to-qubit overhead in fermion time evolution to permutations}

Our aim is to implement a product-formula approximation of the time evolution under a fermionic Hamiltonian $ H$ with minimal circuit depth. For concreteness and simplicity, consider the Hamiltonian $ H=\sum_{i,j=1}^N J_{ij} (a^\dagger_i a_j + \hc)$. We consider a single-step first-order scheme that approximates the time-evolution operator $ U=\exp(-i H t)$ as
\begin{equation}
     U = \prod_{i,j=1}^N \exp(-i t J_{ij} (a^\dagger_i a_j + \hc)) + O(t^2).
    \label{eq:alltoallhopping}
\end{equation}
As we will discuss near the end of this section, both of these restrictions are not necessary: this scheme generalizes to any Trotter scheme and any Hamiltonian (without any restrictions on $k$-locality).

On a fermionic quantum computer, which in particular can directly implement $\exp(-i t J_{ij} (a^\dagger_i a_j + \hc))$  as a single gate, the minimal depth to implement this scheme depends only on how many of the $J_{ij}$ are nonzero and to what extent the structure of the $J_{ij}$ may be parallelized. For instance, for a full-rank tensor, i.e., with $\Theta(N^2)$ non-zero terms, the terms can be grouped such that $O(N)$ of them can be implemented in parallel. This means that the minimal depth is $\Omega(N)$. For a non-full-rank tensor, the situation is more complicated: if there are only $\Theta(N)$ non-zero terms, this may still require depth $\Theta(N)$ if all of the terms have one fermion in common. %
In order to discuss the overhead introduced by the fermion-to-qubit encoding independently of the parallelization restrictions introduced by the Hamiltonian parameters, we define the fermion-to-qubit overhead as the %
multiplicative overhead in gate depth between an implementation on a fermionic quantum computer and a qubit quantum computer, or in other words, the gate depth on a qubit quantum computer divided by the gate depth on a fermion quantum computer, as we have already defined above. To make this comparison, we assume all-to-all connectivity in both qubit and fermion quantum computers.

One of the most efficient ways to perform time evolution in the Jordan-Wigner encoding are the fSWAP networks introduced in Ref.~\cite{kivlichan_quantum_2018}. They have been shown to achieve $O(N)$ depth  (and more generally, $O(N^{k-1}$) for k-local Hamiltonians Ref.~\cite{ogorman_generalized_2019}) for the simulation task in Eq.~\eqref{eq:alltoallhopping} by applying transpositions of neighboring fermions along the Jordan-Wigner chain between fermion gate layers, enabling every possible ordering of the fermions to be traversed with only constant depth overhead. Therefore, fSWAP networks are optimal for all-to-all connected Hamiltonians. However, there is a large class of Hamiltonians with $\Theta(N)$ non-zero $J_{ij}$ that can be implemented in depth $O(1)$ on fermionic quantum computers, but whose implementation with fSWAP networks on qubit quantum computers still uses depth scaling as some power of the system size. One example is the nearest-neighbor Fermi-Hubbard model in two spatial dimensions, for which current fSWAP implementations use $\Theta(\sqrt{N})$ depth~\cite{hemery_measuring_2024}.

In this section, we discuss that, using the schemes introduced in Ref.~\cite{maskara2025} and this work, fSWAP networks can in fact achieve a worst-case upper bound of $O(\log^2 N)$ depth overhead compared to fermionic quantum computers for \emph{general} fermionic Hamiltonians without introducing ancilla qubits. We do so by simply routing the fermion along the Jordan-Wigner chain in-between layers of time evolutions that only act on neighboring fermionic sites, in complete analogy to how routing helps in reducing the depth of all-to-all-connected circuits in 1D nearest-neighbour-connected hardware~\cite{constantinides_optimal_2024}.  

Specifically, for our example, to order $t^2$, 
\begin{align}
 U &\approx \prod_{l=1}^{N_\mathrm{layers}}  U_{\sigma_l}^{-1}  U_{\sigma_l} \prod_{J_{ij}\in \mathrm{layer}\,l} \exp(-i t J_{ij} (a^\dagger_i a_j + \hc))   U_{\sigma_l}^{-1}  U_{\sigma_l}\nonumber \\
&= \prod_{l=1}^{N_\mathrm{layers}}   U_{\sigma_l}^{-1} \prod_{J_{ij}\in \mathrm{layer}\,l} \exp(-i t J_{ij} (a^\dagger_{\sigma_l(i)} a_{\sigma_l(j)} + \hc))  U_{\sigma_l}.\label{eq:UsigmaU}
\end{align}
(We define the product of operators using the following order: $\prod_{i = 1}^N  O_i =  O_N \cdots  O_2  O_1$.) This scheme, see also Ref.~\cite{maskara2025}, is a generalization of the time evolution scheme in the original work introducing fSWAP networks~\cite{kivlichan_quantum_2018} and is also closely related to the double factorization scheme~\cite{motta2021low}, where fermion-overhead-free density-density interactions are alternated with fermionic basis transformations. One way to understand the above scheme is that part of the basis transformation in double factorization is absorbed in the Hamiltonian (such that now it is not purely density-density any more), and the rest of the basis transformation is restricted to permutations of fermions. 

The terms are grouped into $N_\mathrm{layers}$ layers, so that terms within a given layer can be applied in parallel on a fermionic quantum computer. The goal of the permutations is to rearrange the order of the fermion sites in a given layer so that the Hamiltonian terms only act on sites that are adjacent in the Jordan-Wigner encoding, such that $k$-local fermion terms map to $k$-local qubit terms.  In the first line of Eq.~\eqref{eq:UsigmaU}, we insert resolutions of identity in terms of  permutation unitaries $ U_{\sigma_l}$, which apply permutation $\sigma_l$ corresponding to layer $l$: $ U_{\sigma_l} a_i  U_{\sigma_l}^{-1}  = a_{\sigma_l(i)}$.  To enable parallelization of all terms within a given layer in the Jordan-Wigner encoding, the permutation unitaries $ U_{\sigma_l}$ reorder the fermions along the Jordan-Wigner chain such that, for all $J_{ij}$ within a given layer, the fermions $\sigma_l(i)$ and $\sigma_l(j)$ are neighboring along the Jordan-Wigner chain. 
This means that $\exp(-i t J_{ij} (a^\dagger_{\sigma_l(i)} a_{\sigma_l(j)} + \hc)) \rightarrow \exp(-i t J_{ij} ( \Sigma^+_{\sigma_l(i)}  \Sigma^-_{\sigma_l(j)} + \hc))$, where $\Sigma^\pm=( X\pm i Y)/2$ are qubit Pauli operators, can be implemented in depth $O(1)$ and no parallelization restriction arises from the fermion-to-qubit mapping. Furthermore, for every $l < N_\mathrm{layers}$, $ U_{\sigma_{l+1}}  U_{\sigma_l}^{-1}$ can be combined into a single permutation unitary. 

Our discussion above straightforwardly generalizes to any fermionic Hamiltonian (without any restrictions on spacial locality or $k$-locality but with terms that contain an even number of fermionic creation and annihilation operators) as we can in any case, in particular also for an arbitrary order product formula, write time evolution as alternations of fermion permutations and time evolutions in such a way that the time evolution only acts on Jordan-Wigner adjacent modes. For example, for a 4-local term $\propto a^\dagger_i a^\dagger_j a_k a_m$ as appearing in quantum chemistry, the permutations simply exchange four fermion operators such that after the permutations, they are adjacent along the Jordan-Wigner chain and we can map $a^\dagger_{\sigma_l(i)} a^\dagger_{\sigma_l(j)} a_{\sigma_l(k)} a_{\sigma_l(m)}\rightarrow  \Sigma^+_{\sigma_l(i)}  \Sigma^+_{\sigma_l(j)}  \Sigma^-_{\sigma_l(k)}  \Sigma^-_{\sigma_l(m)}$. The resulting at most $N/4$ terms can then be applied without any parallelization restrictions.

As an example, we consider simulation of fermionic Hamiltonians of the form
\begin{equation}
    H= \sum_{i<j} J_{ij} (a^\dagger_i a_j + \hc) + \sum_{i<j} U_{ij} n_i n_j,
\end{equation}
where $n_i=a^\dagger_i a_i$, and $J_{ij},U_{ij}$ are the hopping and interaction coefficient matrices, respectively. This problem is relevant to simulations of materials~\cite{babbush_low-depth_2018}.

\paragraph*{Fermionic fast Fourier transform.} The fermionic fast Fourier transform~\cite{verstraete2009quantum} diagonalizes the hopping Hamiltonian provided $J_{ij}$ is translationally invariant on a $d$-dimensional hypercubic lattice with periodic boundary conditions (we focus on the 1-dimensional case here since the $d$-dimensional FFFT is simply the one-dimensional FFFT applied in all dimensions~\cite{babbush_low-depth_2018}, just as for the usual FFT). Concretely,
\begin{equation}
    \sum_{i<j} J_{ij} (a_i^\dagger a_j + \hc)=\mathrm{FFFT} \left(\sum_i \theta_i n_i \right)\mathrm{FFFT}^\dagger,
\end{equation}
where the $\theta_i$ can be computed classically. This means that the depth of evolving under $\sum_{i<j} J_{ij} (a_i^\dagger a_j + \hc)$ is twice the depth of the FFFT plus the depth of evolving under $\sum_i \theta_i n_i$, which is $O(1)$ in the Jordan-Wigner encoding. 

In the following, we discuss that the FFFT can be applied in polylog depth, implying that a single Trotter step under the hopping Hamiltonian for translationally-invariant systems can be done in polylog depth, as shown in Ref.~\cite{maskara2025}, see also Ref.~\cite{verstraete2009quantum,schuckert_fermion-qubit_2024} for the corresponding statement for fermionic quantum computers.

\begin{corollary}[Ancilla-free polylog-depth FFFT]
    The fermionic fast Fourier transform \cite{Verstraete2009} on $N$ fermionic modes can be implemented in depth $O(\log^2 N)$ in the Jordan-Wigner encoding. \label{thm:polylogFFFT}
\end{corollary}

This means the depth overhead is only $O(\log N)$ as compared to a fermionic quantum computer~\cite{Verstraete2009,schuckert_fermion-qubit_2024}.

\begin{proof}
    Inspecting the circuit for implementing the FFFT (see Fig.~1 in Ref.~\cite{Verstraete2009}) reveals that it consists of $O(\log N)$ $\fSWAP$ layers and $O(\log N)$ layers of parallel interactions. Each interaction layer can be performed in depth $O(1)$ in the Jordan-Wigner encoding because each interaction involves only two modes adjacent in the Jordan-Wigner ordering. In addition, inspecting each $\fSWAP$ layer reveals that each is a staircase permutation, with the exception of the last layer. Because staircase permutations can be performed in depth $O(\log N)$ and general permutations in depth $O(\log^2 N)$ (per  Theorem~\ref{thm:jw-circuit}), the FFFT can be performed in depth $O(\log^2 N)$. 
\end{proof}

Finally, while the focus of this work is on ancilla-free circuits, we can remove one factor of $\log N$ if we allow for ancillas and mid-circuit measurements:

\begin{corollary}[Log-depth FFFT with ancillas~\cite{maskara2025}]
Given any permutation $\sigma\colon [N] \to [N]$, there exists a circuit of depth $O(\log N)$ with $O(N)$ ancillas and measurement and feedforward implementing $\Phi_{\JW}(U_\sigma)$. In addition, there is a circuit of depth $O(\log N)$ implementing the FFFT using $O(N)$ ancillas and measurement and feedforward.
\end{corollary}

\begin{proof}
    The final circuit for performing a fermionic staircase permutation in Figs.~\ref{fig:upgates} and \ref{fig:circuit2_simplified} is composed of three contiguous-range $\CZ$ fanouts, each consisting of two queries to the parity transform and several $\CZ$-fanout operators on independent qubits. $\CZ$-fanout and $P$ (which is equivalent to a $\CNOT$ ladder) can both be implemented in constant depth \cite{B_umer_2025, Jones_2012} using a constant number of rounds of measurements and $O(N)$ ancillas, reducing each of the logarithmically many layers of staircase permutations to constant depth. In addition, as mentioned in Corollary~\ref{thm:polylogFFFT}, each permutation step in the FFFT implementation is a staircase permutation, with the exception of the last, showing that the depth of the entire FFFT circuit is reduced to $O(\log N)$ with $O(N)$ ancillas and measurements.
\end{proof}

Next, we use this efficient FFFT to reduce the circuit depth in important models.

\paragraph*{Fermi-Hubbard-like models.} We first discuss the two-dimensional Fermi-Hubbard model as an example for which our encoding yields a depth advantage when disallowing ancillas. In this scenario, the $J_{ij}$ coefficients are only non-zero for finite-range hopping on a two-dimensional lattice. Specifically, $J_{ij}$ is non-zero only if site $j$ is within a finite number of cartesian units from site $i$ on the 2D lattice.  We also make the same assumption about the $U_{ij}$. In this case, one way to implement Trotter time evolution is to first apply the hopping terms, and then the interactions (or vice versa). Each of those two components is divided into $O(1)$ groups of hopping/interaction  terms that only act on two sites such that all $O(N)$ terms in that group can be applied in parallel. For example, for nearest-neighbor hopping in 2D, there are four groups: two each for the hopping along the two directions of the lattice. This means that, on a fermionic quantum computer, a single Trotter step can be applied in $O(1)$ depth. Compact encodings~\cite{verstraete2005mapping, derby2021} achieve the same scaling by introducing $O(N)$ ancillas and preparing a toric-code state in the beginning of the calculation, which can be done in depth $O(1)$ using midcircuit measurement and feedforward. Therefore, compact encodings achieve $O(1)$ depth overhead.

Our scheme does not achieve $O(1)$ depth, even with $O(N)$ ancillas, and is therefore asymptotically inferior when $O(N)$ ancillas are allowed. However, the lowest-depth ancilla-free scheme known so far uses fSWAP networks based on local transpositions, achieving $O(\sqrt{N})$ depth~\cite{kivlichan_quantum_2018,hemery_measuring_2024}. We can improve on this in cases in which the Hamiltonian is translationally invariant: in that case, we can employ the FFFT to implement time evolution under the hopping part, leading to $O(\log^2 N)$ depth. Because the interaction term can be applied in depth $O(1)$, this implies overall depth $O(\log^2 N)$ for a single Trotter step.

\paragraph*{Materials simulation.} The materials Hamiltonian---or equivalently, quantum chemistry in the plane-wave dual basis~\cite{babbush_low-depth_2018}---is essentially a more complicated version of the Fermi-Hubbard model: the coefficients $J_{ij}$ and $U_{ij}$ are full rank, but both are translationally invariant~\cite{babbush_low-depth_2018}. Therefore, the hopping part can be compressed again to $O(\log^2 N)$ depth without ancillas and $O(\log N)$ depth with ancillas by using the FFFT and polylog-depth permutations. By contrast, because of the unboundedness of the hoppings, no scheme is known how to achieve this depth in the compact encoding. Instead, the depth in the compact encoding defaults to the same asymptotic depth as using fSWAP networks with transpositions~\cite{kivlichan_quantum_2018} as the same geometric restrictions apply. A scheme to reduce the depth to $O(\log N)$ was proposed in Ref.~\cite{schuckert_fermion-qubit_2024}, but required $O(N^2)$ ancillas. For the interaction term, no scheme is known to apply it in polylog depth without ancillas; however, an $O(\log N)$-depth scheme is known using $O(N \log N)=\tilde O(N)$ ancillas~\cite{low_hamiltonian_2019}. Therefore, the schemes proposed in Ref.~\cite{maskara2025} and in this work achieve polylog depth for a single Trotter step, as already noted in ~\cite{maskara2025}. Finally, for obtaining the end-to-end depth for simulating time evolution until time $T$, we need to multiply this by the number of Trotter steps. The best known upper bound is $O(NT)$ from Ref.~\cite{childs2021theory}, assumes arbitrarily large Trotter order and fixed total simulation error. Therefore, our schemes achieve $\tilde O(N T)$ depth complexity with $\tilde O(N)$ qubits. This quasilinear-in-number-of-plane-waves scaling of qubits and depth has to our knowledge not been achieved for qubit quantum computers yet and this scaling is on par (asymptotically) with fermionic quantum computers~\cite{schuckert_fermion-qubit_2024}.

\begin{table*}[t]
\renewcommand{\arraystretch}{1.2}
\begin{tabular}{|l|cc|cc|cc|cc|cc|}
\hline
\multirow{2}{*}{\textbf{ (Sub-)routine}} 
 & \multicolumn{2}{c|}{\textbf{Compact}~\cite{verstraete2005mapping, derby2021}} 
 & \multicolumn{2}{c|}{\textbf{fSWAP}~\cite{kivlichan_quantum_2018}} 
 & \multicolumn{2}{c|}{\textbf{pfSWAP}~\cite{schuckert2025faulttolerant}} 
 & \multicolumn{2}{c|}{\textbf{Maskara et al.}~\cite{maskara2025}} 
 & \multicolumn{2}{c|}{ \textbf{This work}}\\
 & Depth & Anc. & Depth & Anc. & Depth & Anc. & Depth & Anc. & Depth & Anc. \\
\hline
Fermion permut.
 & $N$ & $N$
 & $N$ & $0$
 & $1$ & $N^{2}$
 & $\log N$ & $N$ 
 & $\log^{2} N$ & $0$\\

FFFT
 & $N\log N$ & $N$
 & $N\log N$ & $0$
 & $\log N$ & $N^{2}$
 & $\log N$ & $N$ 
 & $\log^{2} N$ & $0$\\
\hline
Materials (1 step)
 & $N$& $N$
 & $N$ & $0$
 & $\log N$ & $N^{2}$
 & $\log N$ & $N\log N$ 
 & $\log^{2} N$ & $N\log N$\\

Materials (full)
 & $N^2$ & $N$
 & $N^{2}$ & $0$
 & $N\log N$ & $N^{2}$
 & $N\log N$ & $N\log N$ 
 & $N\log^{2} N$ & $N\log N$\\

\shortstack{2D Fermi-Hubbard}
 & $1$ & $N$
 & $\sqrt{N}$ & $0$
 & $\log N$ & $N^{2}$
 & $\log N$ & $N$ 
 & $\log^{2} N$ & $0$\\
\hline
\end{tabular}
\caption{Asymptotic circuit depth and ancilla count for key fermionic subroutines and applications. All quantities are $O(\cdot)$ upper bounds.
Compact refers to any compact encoding, e.g., Verstraete-Cirac~\cite{verstraete2005mapping} or Derby-Klassen~\cite{derby2021}; fSWAP is the transposition-based fSWAP network from Ref.~\cite{kivlichan_quantum_2018}; 
pfSWAP refers to parallel fSWAP from Ref.~\cite{schuckert_fermion-qubit_2024} using $O(N^2)$ ancillas to brute-force parallelize the CZ circuits in the fSWAP networks; ``2D Fermi-Hubbard'' applies, more generally, to any model with $O(1)$-range hopping and interactions and assumes periodic boundary conditions for pfSWAP and our work (as we assume the usage of the FFFT for the kinetic term), and the gate counts are for a single Trotter step. The number of Trotter steps $O(N)$ used in ``Materials (full)'' assumes arbitrarily large Trotter order, fixed evolution time, and fixed simulation error. Maskara et al.~\cite{maskara2025} makes use of ancillas in the permutation step, whereas our work only uses ancillas for parallelizing the interaction part of the Hamiltonian~\cite{low_hamiltonian_2019}.}
\label{tab:comparison}
\end{table*}

\section{Discussion and Outlook \label{sec:outlook}}
In this work, we build on the recent exponential improvement in fermionic routing introduced by Maskara et al.~\cite{maskara2025}, who showed that the worst-case routing overhead in fermion-to-qubit mappings can be reduced from $O(N)$ to $O(\log N)$ by decomposing general fermionic permutations into $O(\log N)$ constant-depth interleaves and using $\Theta(N)$ ancilla qubits and mid-circuit measurements to compress the CZ circuits part of fermi-SWAP networks~\cite{kivlichan_quantum_2018}. We rederive and generalize their results using an alternative construction. Moreover, we show that similar exponential improvements persist even in the ancilla-free setting: any fermionic permutation in the Jordan-Wigner encoding can be implemented in depth $O(\log^2 N)$ without ancillas. We further extend our analysis to arbitrary product-preserving ternary-tree encodings by constructing efficient mappings between those encodings and the Jordan-Wigner encoding in $O(\log^2 N)$ depth. This establishes that fermionic routing can be efficiently performed in a broad class of fermionic encodings. As a result, key primitives such as the fermionic fast Fourier transform (FFFT) and single-Trotter-step material simulations can be executed in polylogarithmic depth using only $\tilde O(N)$ qubits for $N$ fermionic modes.

Our results assume hardware with all-to-all connectivity, but they also imply upper bounds for limited-connectivity devices. In particular, the worst-case overhead for fermionic simulation in constrained architectures is $O(\log^2 N)$ times the worst-case depth required to route qubits on the underlying hardware graph, if that hardware graph has $\Omega(N)$ disjoint edges~\cite{yuan2025full}. For instance, qubit systems with grid connectivity can implement the FFFT in depth $O(\sqrt{N}\log^2 N)$; current neutral-atom architectures achieve $O(\sqrt{N}\log^3 N)$ depth, while upgraded neutral-atom hardware could reach $O(\log^3 N)$ depth~\cite{constantinides_optimal_2024}. This represents an exponential improvement over the prior best-known FFFT algorithm, which achieved $O(N\log N)$ depth~\cite{babbush_low-depth_2018} on a linear geometry, with no mechanism to exploit higher connectivity. The above-cited depths in the presence of hardware constraints is likely too pessimistic as it neglects additional structure in the circuits. In fact, the FFFT is an example where that structure can be exploited, and Maskara et al.~\cite{maskara2025} showed several others.

Finally, these results motivate further study of fermionic routing and simulation on hardware with restricted connectivity, particularly the trade-offs among ancilla count, mid-circuit measurement and feedforward, and depth overhead. Beyond asymptotic scaling, a detailed analysis of constant factors in depth and gate count will be essential for identifying the most practical regimes of advantage. It also remains an open question whether analogous circuit-structural insights could lead to improved boson-to-qubit encodings~\cite{sawaya2020resourceefficient, crane2024hybrid, liu2025hybrid}.

\section*{Acknowledgments}
We thank Maskara et al.\ for sharing their results ahead of posting their manuscript~\cite{maskara2025} on the arXiv.

N.C., J.Y., D.D., A.F., A.S., and A.V.G.\ acknowledge support by  the U.S.~Department of Energy, Office of Science, National Quantum Information Science Research Centers, Quantum Systems Accelerator (QSA). %
A.M.C., N.C., J.Y., D.D., M.J.G., A.F., A.S., and A.V.G.\ acknowledge support by NSF QLCI (award No.~OMA-2120757), DoE ASCR Quantum Testbed Pathfinder program (awards No.~DE-SC0019040 and No.~DE-SC0024220), and the U.S.~Department of Energy, Office of Science, Accelerated Research in Quantum Computing, Fundamental Algorithmic Research toward Quantum Utility (FAR-Qu).
N.C., J.Y., D.D., A.F., A.S., and A.V.G. acknowledge support by ONR MURI, NSF STAQ program, AFOSR MURI, DARPA SAVaNT ADVENT, ARL (W911NF-24-2-0107), and NQVL:QSTD:Pilot:FTL.

\end{document}